\documentclass[letterpaper, 11pt]{article}
\usepackage{times}
\usepackage{helvet}
\usepackage{courier}
\usepackage{authblk}
\usepackage{appendix}
\usepackage{fullpage}
\usepackage[utf8]{inputenc} 
\usepackage[T1]{fontenc}    
\usepackage{hyperref}       
\usepackage{url}            
\usepackage{booktabs}       
\usepackage{amsfonts}       
\usepackage{nicefrac}       
\usepackage{microtype}      
\usepackage{enumerate} 
\usepackage[cmex10]{amsmath} 
\allowdisplaybreaks[4] 
\usepackage{amssymb} 
\usepackage{amsthm} 
\usepackage{graphicx}
\usepackage{caption}
\usepackage{subcaption}
\usepackage{fancybox}
\usepackage[ruled,vlined]{algorithm2e}
\usepackage{xcolor}
\usepackage{epstopdf}
\newtheorem{thm}{Theorem}
\newtheorem{lemma}{Lemma}

\newcommand{\tree}{T}

\newcommand{\obsTime}{t}
\newcommand{\ecce}{e}
\newcommand{\edges}{{\cal E}}
\newcommand{\nodes}{{\cal V}}
\newcommand{\infObs}{{\cal I}'}

\newcommand{\candidateSet}{{\cal K}}
\newcommand{\dist}{d}

\newcommand{\levelset}{{\cal L}}
\newcommand{\singleLevelSet}{{\cal L}'}

\newcommand{\graph}{g}
\newcommand{\collisionset}{{\cal R}}
\newcommand{\collisionsize}{R}

\newcommand{\offspring}{\Phi}
\newcommand{\offspringsize}{\phi}
\newcommand{\sourceSet}{{\cal S}}
\newcommand{\neighbors}{{\cal N}}
\newcommand{\estimatorSet}{{{\cal S}^\prime}}

\DeclareMathOperator*{\argmin}{\arg\!\min}
\DeclareMathOperator*{\argmax}{\arg\!\max}
\title{Catch'Em All: Locating Multiple Diffusion Sources in Networks with Partial Observations}

\author[1]{Kai Zhu\thanks{zhukai93@gmail.com} }
\author[2]{Zhen Chen\thanks{Zhen.Chen.1@asu.edu}}
\author[2]{Lei Ying\thanks{Lei.Ying.2@asu.edu}}
\affil[1]{Google Inc.}
\affil[2]{School of of Electrical, Computer and Energy Engineering,
Arizona State University}

\begin{document}
\date{}
\maketitle
\begin{abstract}
This paper studies the problem of locating multiple diffusion sources in networks with partial observations. We propose a new source localization algorithm, named Optimal-Jordan-Cover (OJC). The algorithm first extracts a subgraph using a candidate selection algorithm that selects source candidates based on the number of observed infected nodes in their neighborhoods. Then, in the extracted subgraph, OJC finds a set of nodes that ``cover'' all observed infected nodes with the minimum radius. The set of nodes is called the Jordan cover, and is regarded as the set of diffusion sources.  Considering the heterogeneous susceptible-infected-recovered (SIR) diffusion in the Erd\H os-R\' enyi (ER) random graph, we prove that OJC can locate all sources with probability one asymptotically with partial observations. OJC is a polynomial-time algorithm in terms of network size. However, the computational complexity increases exponentially in $m,$ the number of sources. We further propose a low-complexity heuristic based on the K-Means for approximating the Jordan cover, named Approximate-Jordan-Cover (AJC). Simulations on random graphs and real networks demonstrate that both AJC and OJC significantly outperform other heuristic algorithms.
\end{abstract} 

\section{Introduction}
Diffusion source localization (or called the information source detection) is to infer the source(s) of an epidemic diffusion process in a network based on some observations of the diffusion. Possible observed information includes node states (e.g., infected or susceptible) and the timestamps at which nodes changed their states. The solution to this problem has a wide range of applications. In epidemiology, identifying patient zero helps diagnose  the cause and the origin of the disease. For cybersecurity, tracing the source of malware is an important step in the investigation of a cyber attack. On online social networks, the trustworthiness of news/information heavily depends on its source.

Since the seminal work of Shah and Zaman \cite{ShaZam_10}, the problem has received a lot of attention from different research communities, including machine learning, signal processing and information theory (a detailed discussion can be found in the related work). However, most existing results assume that the diffusion is from single source and the observed information is one or multiple complete network snapshot(s). Furthermore, provably guarantees on the detection rate have only been established for tree networks except a recent paper by Zhu and Ying  \cite{ZhuYin_16}, where they proposed a Short-Fat Tree (SFT) algorithm and proved that under the single-source, homogeneous Independent-Cascade (IC) model, SFT locates the actual source  in the Erd\H os-R\' enyi random graph \cite{ErdRen_59} with probability one asymptotically (as the size of the network increases) when a complete snapshot of the network is given.

In this paper, we consider the diffusion source localization problem in a setting that generalizes the existing ones at several important directions.
\begin{itemize}
\item {\em Multiple sources versus single source:} In this paper, the diffusion can be originated from multiple nodes simultaneously, instead of from a single source. When the infection duration is sufficiently short, the infected subnetworks from different sources are disconnected components. In such cases, the single-source localization algorithms can be applied to each of the infected subnetwork. We, however, do not make such an assumption, and consider the scenario where the infected subnetworks may overlap with each other, so the single-source localization algorithms cannot be directly applied.

\item {\em A partial snapshot versus a complete snapshot:} In this paper, we assume a partial snapshot in which each node reports its state with some probability, which is in contrast to a complete snapshot assumed in the literature where all nodes' states are observed. Because of a partial snapshot, the sources may not report their states and be observed as infected nodes; and the observed infected nodes may not form a connected component. Both increase the uncertainty and complexity of the problem. In fact, it turns out to be critical to have a candidate selection algorithm to select source candidates from unobserved nodes but only use observed infected nodes in computing the infection eccentricity. The selection step yields $27\times$ reduction on the computing time in our simulations while guaranteeing the same the detection rate, and yields $600\times$ reduction on the computing time with a slight reduction of the detection rate. 

\item {\em Heterogeneous diffusion versus homogeneous diffusion:} Our algorithm applies to the heterogeneous SIR diffusion model where links have different infection probabilities and nodes have different recovery probabilities. The asymptotic guarantees on the detection rate hold for the heterogeneous SIR model.
\end{itemize}

While some of these extensions have been investigated in the literature individually, our model includes all three extensions. We propose a novel algorithm for locating multiple sources for such a general model and prove theoretical guarantees on the detection rate for non-tree networks. The main results of the paper are summarized below.
\begin{itemize}
\item[(1)] We introduce the concept of Jordan cover, which is an extension of Jordan center.  Loosely speaking, a Jordan cover with size $m$ is a set of $m$ nodes that can reach all {\em observed} infected nodes with the minimum hop-distance. We propose Optimal-Jordan-Cover (OJC), which consists of two steps:  OJC first selects a subset of nodes as the set of the candidates of the diffusion sources; and then it finds a Jordan cover in the subgraph induced by the candidate nodes and the observed infected nodes. We emphasize that only the hop-distance to the observed infected nodes is considered in computing a Jordan cover.

\item[(2)] We analyze the performance of OJC on the ER random graph, and establish the following performance guarantees.
\begin{itemize}
\item[(i)] When the infection duration is shorter than $\frac{2}{3} \frac{\log n}{\mu},$  where $\mu$ is the average node degree and $n$ is the number of nodes in the network, OJC identifies the sources with probability one asymptotically as $n$ increases.
\item[(ii)] When the infection duration is at least $\left\lceil\frac{\log n}{\log \mu + \log q}\right\rceil +2$ where $q$ is the minimum infection probability, {\em under any source location algorithm,} the detection rate diminishes to zero as $n$ increases under the Susceptible-Infected (SI) and Independent-Cascade (IC) models, which are special cases of the SIR model.

\end{itemize}
\item[(3)] The computational complexity of OJC is polynomial in $n$, but exponential in $m.$  We further propose a heuristic based on the K-Means for approximating the Jordan cover, named Approximate-Jordan-Cover (AJC). Assuming a constant number of iterations when using the K-Means, the computational complexity of AJC is $O(n E),$ where $E$ is the number of edges. Our simulations on random graphs and real networks demonstrate that both AJC and OJC significantly outperform other heuristic algorithms.
\end{itemize}

\subsection{Related Work}

Shah and Zaman \cite{ShaZam_10,ShaZam_11} studied the rumor source detection problem, and proposed a new graph centrality called \emph{rumor centrality}. They proved that the node with the maximum rumor centrality is the maximum likelihood estimator (MLE) of the diffusion source on regular trees under the continuous-time SI model. In addition, it has been proved that the detection probability is greater than zero on regular trees and approaches one for geometric trees. \cite{ShaZam_12} analyzed the detection probability of rumor centrality for general random trees. Later, the performance of rumor centrality has been studied under different models, including multiple sources \cite{LuoTayLen_13}, incomplete  observations \cite{KarFra_13}, multiple independent diffusion processes from the same source \cite{WanDonZha_14}.

Kai and Ying \cite{ZhuYin_13} proposed a sample path based approach for single source detection under the SIR model and introduced the concept of \emph{Jordan infection center}. Assuming the homogeneous SIR model, a complete network snapshot and regular tree networks, they \cite{ZhuYin_13} proved that the Jordan infection center is the root of the most likely diffusion sample path, and it is within a constant hop-distance from the actual source with a high probability. {\em Assuming tree networks}, the performance of the Jordan infection center has been further studied, including partial observations under the heterogeneous SIR model \cite{ZhuYin_14}, multiple sources \cite{CheZhuYin_14}, the SI model and SIS model \cite{LuoTayLen_14}.

Besides the rumor centrality and sample path based approach, diffusion source localization has also been investigated from other perspectives: 1) \cite{LapTerGun_10} proposed a dynamical programming algorithm for the IC model; 2) A variation of eigen centrality was proposed in \cite{PraVreFal_12} to detect multiple sources under the SI model; 3) \cite{LokMezOht_14} derived the mean field approximation of the MLE of the actual source and proposed a dynamic message passing algorithm based on that. Furthermore, \cite{PinThiVet_12,ZhuCheYin_15,AgaLu_13,ZejGomSin_13,FarGomZam_15} have proposed algorithms to identify the source with partial timestamps.

In this paper, we propose two new source localization algorithms, OJC and AJC, based on the Jordan cover. OJC guarantees probability one detection asymptotically  on the ER random graph under the multi-source, heterogeneous SIR model, and with an incomplete snapshot, which is the first multi-source localization algorithm with provable guarantees for non-tree networks. AJC is a low-complexity approximation of OJC.

\section{Problem Formulation}\label{sec:ProblemFormulation}
We assume the network is represented by an undirected graph $g$. Denote by $\edges(g)$ the set of edges and $\nodes(g)$ the set of nodes in graph $g$. Let $n$ denote the number of nodes and $E$ denote the number of edges. We further assume a heterogeneous SIR model for diffusion. In this model, each node has three possible
        states: susceptible (S), infected (I) and recovered (R). Time is slotted. At the beginning of each time slot, each
        infected node (say node $u$) attempts to infect its neighbor (say node $v$) with
        probability $q_{uv},$ independently across edges.  We call $q_{uv}$ the \emph{infection
        probability} of edge $(u,v)$. At the end of each time slot, each infected node (say node $u$) recovers with
        probability $r_{u},$ independent of other infected nodes. We call $r_{u}$ the \emph{recovery probability}. We further assume  $q_{uv}\in(0,1]$ for all edges $(u,v)\in \edges(g)$ and $r_{v}\in[0,1]$ for all nodes $v\in\nodes(g)$.

Note that the SIR model includes two important special cases. When the recovery probability is zero, the
SIR model becomes the Susceptible-Infected (SI) model \cite{Bai_75}, where infected nodes cannot recover. When the recovery probability is one, the SIR model becomes the Independent Cascade (IC)
model \cite{GolLibMul_01} by regarding both infected nodes and recovered nodes as active nodes and regarding the susceptible nodes as inactive nodes.

We assume the epidemic diffusion starts from $m$ sources in the network. In other words, at time slot
$0$, $m$ nodes (sources) are in the infected state and all other nodes are in the susceptible state.
Denote by $s_1,s_2,\cdots,s_m$ the sources and $\sourceSet$ the set of
sources, i.e., ${\cal S}=\{s_1,s_2,\cdots,s_m\}$. We assume $m$ is {\em a constant independent of $n$.}

Finally, we assume that a {\em partial} snapshot of the network state at time slot $t$ is given, with an unknown observation time $t.$ In the snapshot, each infected or
recovered node reports its state with probability $\theta_v\in(0,1)$, independent of other nodes. If a node reports its state, we call it an observed
node. Denote by $\infObs$ the set of observed infected and recovered nodes. In this paper, we call infected
nodes and recovered nodes as ``infected nodes'' unless explicitly clarified.

Based on $\infObs$, the source localization problem is to find $\sourceSet$ that solves the following maximum likelihood (ML) problem
$$
    {\cal W}^* = \argmax_{{\cal W}\subset \nodes(g)} \Pr(\sourceSet = {\cal W}|\infObs).
$$ Even with a single diffusion source, this problem is known to be a difficult problem \cite{ShaZam_10,ZhuYin_13} on non-tree networks. Therefore, instead of solving the ML problem above, we are interested in algorithms with {\em asymptotic perfect detection}, i.e., finding all sources with probability one as the network size increases. We believe this alternative metric is reasonable because we often need to solve the problem for large-size networks such as online social networks, and an algorithm with asymptotic perfect detection can detect sources with a high probability when the network size is large. 

\section{Algorithms}\label{sec:Algorithm}

\begin{algorithm}
	\SetAlgoLined
	
	\KwIn{$\infObs, g, Y$;}
	\KwOut{$g^-$ (the candidate subgraph)}
	
	Set ${\cal K}$ to be an empty set.
	
	\For {$v \in \nodes(g)$}{
		\If{$|\neighbors(v)\cap\infObs|>= Y$}{
			Add $v$ to ${\cal K},$ where $\neighbors(v)$ is the set of neighbors of node $v.$
		}
	}
	
	Set ${\cal K}^+$ to be ${\cal K}\cup \infObs$.
	
	Set $g'$ to be the graph induced by set ${\cal K}^+$.
	
    \color{black}
    Find all connected components in $g'.$
	
			\uIf{$g'$ is connected}{
				Set $g^- = g'$.
			}
			\Else{
				Randomly select one node in each components of $g'$. Denote by ${\cal R}$ the set of the selected nodes.
				
				Randomly select one node $v \in {\cal R}.$
				
				\For{ $u\in {\cal R}\backslash v$} {
					Compute the shortest path $P$ from $v$ to $u.$
					
					Set $g'=g'\cup P.$
					}
				Set $g^- = g'$
				}
	\color{black}
	\Return $g^-, {\cal K}$.
	\caption{The Candidate Selection Algorithm}\label{alg:CS}
\end{algorithm}

In this section, we present OJC and AJC based on the concept of Jordan cover.

Define the hop-distance between a node $v$ and a {\em node set} ${\cal W}$ to be the minimum hop-distance between node $v$ and any node in ${\cal W},$ i.e., $$d(v,{\cal W}) \triangleq \min_{u\in{\cal W}} d(v,u).$$ We then define the \emph{infection eccentricity} of node set ${\cal W}$ to be the maximum hop-distance from an infected node in $\infObs$ to set $\cal W,$ i.e.,
\begin{equation}
\ecce({\cal W}, {\cal \infObs}) = \max_{v \in \infObs} d(v,{\cal W}).\label{ns: ecc}
\end{equation}  
We further define $m$-Jordan-cover ($m$-JC) to be the set ${\cal W}^*(\candidateSet,\infObs, m)$ such that
\begin{align}
{\cal W}^*(\candidateSet,\infObs, m) = \argmin_{{\cal W}\in \{{\cal W}||{\cal W}| = m, {\cal
W}\subset \candidateSet\}} \ecce({\cal W},\infObs).\label{eqn:JC}
\end{align}
where $\infObs$ is the set of observed infected nodes that $m-$JC needs to cover and $\candidateSet$ is the \emph{candidate set} for the sources. Therefore, $m$-JC is the set of $m$ nodes in $\candidateSet$ with the minimum infection eccentricity.  

We now introduce the \emph{optimal Jordan cover} (OJC) algorithm whose asymptotic detection rate will be analyzed in Section \ref{sec:MainResult-OJCperformances}.

\noindent{\bf The Optimal Jordan Cover (OJC) Algorithm}
\begin{itemize}
    \item {\bf Step 1: Candidate Selection:} Let $Y$ be a positive integer.  The candidate set $\cal K$ is the set of nodes with more than $Y$ observed infected neighbors. In addition, define ${\cal
K}^+\triangleq {\cal K}\cup \infObs.$  Denote by $g^-$ a connected subgraph of $g$ \emph{induced} by
node set ${\cal K}^+$. An {induced graph} is a subset of nodes of a graph with
all edges whose endpoints are both in the node subset. If the induced graph is not connected, we select a random node in each component, randomly pick one selected node and add the shortest pathes from this node to all other selected nodes to form a connected $g^-.$ We call $Y$ the \emph{selection
threshold}.  The pseudo code of the candidate selection algorithm for selecting $\cal K$ and $g^-$ can be found in Algorithm \ref{alg:CS}.

\item {\bf Step 2: Jordan Cover}: For any $m$ combination of nodes in ${\cal K}$ in Step 1, we compute
        the infection eccentricity of the node set as defined in (\ref{ns: ecc}) on subgraph $g^-,$ and select the combination with the minimum infection eccentricity as the set of sources. Ties are broken by the total distance from the observed infected to the node set, i.e., $\sum_{v\in{\cal I}'} d(v, {\cal W}).$ 
\end{itemize}

With a properly chosen threshold $Y,$ the candidate selection step includes all sources in $\cal K$ with a high probability and excludes nodes that are more than $t+1$ hops away from all sources. By limiting the computation on the induced subgraph $g^-,$ the computational complexity is reduced significantly. From simulations, we will see that it results in $27\times$ reduction of the running time without affecting the detection rate. The asymptotic detection rate of OJC will be studied in Theorem \ref{thm:JIC}. Under some conditions, OJC
identifies all sources with probability one asymptotically.

OJC is a polynomial-time algorithm for given $m$, but the complexity increases exponentially in $m.$ To further reduce the complexity, we propose Approximate Jordan Cover (AJC), which replaces Step 2 of OJC with the K-Means algorithm \cite{HarWon_79}. As shown in the simulations, the performance of the AJC algorithm, in terms of both detection rate and the error distance, is close to OJC with much shorter running time. The computational complexity of both algorithms are summarized in the following theorem. 
\begin{thm}
The computational complexity of OJC is \begin{align*}
O\left(|{\cal I}'| \left(|\edges(g)|+m{|\nodes(g^-)| \choose m}\right)\right),
\end{align*} and the computational complexity of AJC is
\begin{align*}
O\left(|{\cal I}'|\left(|\edges(g)|+H|\nodes(g^-)| \right)\right),
\end{align*} where $H$ is the number of iterations used in the K-Means algorithm in AJC.
\end{thm}
\begin{proof}
We first analyze the complexity of the OJC algorithm. For simplicity, denote by $V = |\nodes(g)|$ the number of nodes and $E = |\edges(g)|$ the number of edges.

In the candidate selection stage, each node needs to compute its degree. Therefore, each edge is counted twice and each node is processed once. The complexity is $O(V + E).$ Counting the number of the connected components using breadth first search is of complexity $O(V + E)$. When the induced graph is not connected, the complexity to compute the shortest paths from one node to all other nodes in an unweighted graph is of complexity $O(V+E)$. As a summary, the complexity of the candidate selection algorithm is $O(V + E)$.

The resulting subgraph has $|\nodes(g^-)|$ nodes and at most $E$ edges. Next, we compute the complexity of the OJC.

We first compute the distances from nodes in $g^-$ to nodes in ${\cal I}'$. Note this is equivalent to do a breadth-first search from each node in ${\cal I}'$. The complexity of the BFS is $O(|\nodes(g^-)|+E).$ Therefore, the complexity for this step is $O(|{\cal I}'|(|\nodes(g^-)|+E)).$ The results are saved in a two dimensions hashtable so that querying the distance from one observed infected node and another node in graph $g^-$ is of complexity $O(1).$

After the above precomputation, for each set of nodes with size $m,$ we want to obtain its infection eccentricity. For each observed infected node, we query the hash table to find the minimum distance to the set of nodes with size $m.$ The complexity is $O(m).$ To compute the infection eccentricity of one set is of complexity $O(m|{\cal I}'|).$ In addition, there are ${|\nodes(g^-)| \choose m}$ possible node sets. Therefore, the complexity is $O((m|{\cal I}'|){|\nodes(g^-)|\choose m})$

As a summary, the complexity of the OJC algorithm is
\begin{align*}
&  O\left(V+E+|{\cal I}'|(|\nodes(g^-)|+E)+m|{\cal I}'|{|\nodes(g^-)|\choose m}\right) \\
=& O\left(|{\cal I}'|\left(E+m{|\nodes(g^-)|\choose m}\right)\right)
\end{align*}

Next, we analyze the complexity of the AJC algorithm. Note the complexity of the candidate selection and the precomputation are the same.  The complexity of the Kmeans algorithm is analyzed as follows.

The complexity of the membership assignment phase is $O(m|{\cal I}'|)$. For each observed infected node, we only need to query its distance to the preselected $m$ sources and each query is of complexity $O(1)$ based on the results of the precomputation.

For the center update phase, for each cluster, we need to search all $|\nodes(g^-)|$ nodes to find a center. Therefore, the complexity is $|\nodes(g^-)||{\cal I}'|.$

As a summary, the complexity of one iteration of the Kmeans algorithm is
\[
O\left(\left(m+|\nodes(g^-)|\right)|{\cal I}'|\right)
\]

Denote by $H$ the number of iterations, the complexity of the AJC algorithm is
\begin{align*}
&O\left(V+E+|{\cal I}'|(|\nodes(g^-)|+E)+N\left(\left(m+|\nodes(g^-)|\right)|{\cal I}'|\right)\right)\\
=& O\left(|{\cal I}'|\left(E+H|\nodes(g^-)|\right)\right)
\end{align*}
\end{proof}
\section{Asymptotic Analysis of OJC}\label{sec:MainResults}
In this section, we present the asymptotic analysis of the detection rate of OJC on the ER random graph. The results include the conditions that guarantee probability one detection  and the conditions under which it is impossible to detect the source set with nonzero probability under any source localization algorithm.
\subsection{Asymptotic perfect detection on the ER random graph}\label{sec:MainResult-OJCperformances}
We first present the positive result that shows that on the ER random graph, OJC identifies the $m$ sources with probability one asymptotically under some conditions. Recall that $n$ is the number of nodes in the graph, and $m$ is the number of sources, which is a constant independent of $n.$ Denote by $p$ the wiring probability of the ER random graph, which is the probability that there exists a link between two nodes. Let $\mu = np$, which is the average node degree. Define $q \triangleq \min_{u,v \in \nodes(g)} q_{u,v},$ i.e., the minimum infection probability over all edges and $\theta \triangleq \min_{v \in \nodes(g)} \theta_v$
i.e., the minimum report probability over all nodes.

\begin{thm}\label{thm:JIC}
OJC identifies all $m$ sources with probability one as $n\rightarrow \infty$ when the following conditions hold:
\begin{itemize}
\item[(c1):] $\mu q \theta=\Omega\left(\log n\right),$\footnote{Throughput this paper, the asymptotic order notation is defined for $n\rightarrow \infty.$}

\item[(c2):] $\lim\sup_{n\rightarrow\infty}\frac{Y}{\mu q \theta}<1$ and  $\lim\inf_{n\rightarrow\infty}\frac{Y}{\mu q \theta}>0,$ and

\item[(c3):]  $\obsTime=\omega(D)$ and  $\lim\sup_{n\rightarrow \infty} \frac{\obsTime}{ \frac{\log n}{\log\mu}}<\frac{2}{3}.$ 
\end{itemize}
\end{thm}

\begin{proof}
 In this proof, we will show that one source $s_i \in {\cal W}^*(\candidateSet,\infObs, m)$ with a high probability. Then with a union bound, we will show that
\[
 {\cal S} = {\cal W}^*(\candidateSet,\infObs, m).
\]
with a high probability for sufficiently large $n$.

Recall that we regard the recovered nodes and infected nodes as "infected" since the recovery process is not related to the proof. Without loss of generality, we consider $s_1$. Throughout the proof, we consider the BFS tree rooted at $s_1$. In particular, the level of one node means the level of the node on the BFS tree rooted at
$s_1$.

We first introduce and recall some necessary notations terms.

For an ER random graph $g$.
\begin{itemize}
\item A node $v$ is said to be on level $i$ if $d_{s_1v} = i$. Denote by $\levelset_i$ the set of nodes from level 0 to level $i$ and $l_i = |\levelset_i|.$

\item Denote by $L'_i$ the set of nodes on level $i$. In addition, $l'_i$ is the number of nodes on level $i$.
\item The descendants of node $v$ in a tree are all the nodes in the subtree rooted at $v.$ In addition, $v$ is the ancestor of all its descendants.
\item The offsprings of a node on level $k$ (say $v$) are the nodes which are on level $k+1$ and have edges to $v.$ Denote by $\offspring(v)$ the offspring set of $v$ and $\offspringsize(v)=|\offspring(v)|.$
\item Denote by $p$ the wiring probability in the ER random graph.
\item Denote by $n$ the total number of nodes.
\item Denote by $\mu=np.$
\item Denote by $\hbox{Bi}(n,p)$ the binomial distribution with $n$ number of trials and each trial succeeds with probability $p.$
\item Denote by $\tree^\dag$ the BFS tree rooted at $s_1.$
\item Denote by $\offspring'(v)$ the set of offsprings of node $v$ on $\tree^\dag$ and $\offspringsize'(v)=|\offspring'(v)|.$
\item Denote by $\graph_\obsTime$ the subgraph induced by all nodes within $\obsTime$ hops from $s$ on the ER graph.  The \emph{collision edges} are the edges which are not in $\tree^\dag$ but in $\graph_\obsTime,$ i.e., $e\in \edges(\graph_\obsTime)\backslash\edges(\tree^\dag).$
\item A node who is an end node of a collision edge is called a \emph{collision node}.  Denote by $\collisionset_k$ the set of collision edges whose end nodes are within level $k$ and $\collisionsize_k=|\collisionset_k|.$

\item Denote by ${\cal Z}_i$ the set of nodes which are infected at time $i$.

\item Denote by $\tilde{\cal Z}^i_j(v)$ the set of nodes that
    are infected at time slot $i$, on level $j$, when $s_1$ is the only infection source in the
    graph and the descendants of node $v$ in the BFS tree rooted at $s_1$ and $\tilde{\cal
    Z}'^i_j(v)$ are the observed infected nodes in $\tilde{\cal Z}^i_j(v)$. $\tilde{Z}^i_j(v)$ and
    $\tilde{Z}'^i_j(v)$ are defined as the cardinality respectively.
\item Denote by $\psi(v)$ the number of observed infected neighbors of node $v$.
\item Denote by $\psi'(v)$ the number of infected offsprings of node $v$ (the offspring is defined based
on the BFS tree rooted at node $s_1$).
\item Denote by $\psi''(v)$ the number of
\emph{observed} infected offsprings of node $v$.
\end{itemize}

To prove $s_1\in {\cal W}^*(\candidateSet,\infObs, m),$ we need to show that any set ${\cal W}$ such
that $s_1 \not\in {\cal W}$ has a infection eccentricity larger than $t$ on $\graph^-$.  We need the following asymptotic high probability events.
\begin{itemize}
\item {\bf Offsprings of each node.} Consider the BFS tree rooted at source $s_1$. Define
\[
E_1=\{\forall v \in \levelset_{\obsTime + D}, \phi'(v)\in((1-\delta)\mu,(1+\delta)\mu)\}.
\]
$E_1$, when occurs, provides upper and lower bounds for the number of offsprings of each node in
$\levelset_{\obsTime + D}$.

\item {\bf Total number collision edges.} Consider the BFS tree rooted at source $s_1$. We define event $E_2$ when the following upper bound on the collision edges holds
\[
\collisionsize_j
\begin{cases} =0 & \mbox{if } 0<j\leq \lfloor m^-\rfloor, \\
\leq 8\mu & \mbox{if } \lfloor m^-\rfloor<j<\lceil m^+\rceil,\\
\leq \frac{4[(1+\delta)\mu]^{2j+1}}{n} & \mbox{if } \lceil m^+\rceil\leq j\leq \frac{\log n}{(1+\alpha)\log \mu}.  \end{cases} \]
where $m^+=\frac{\log n}{2\log[(1+\delta)\mu]}$ and $m^-=\frac{\log n-2\log \mu-\log
8}{2\log[(1+\delta)\mu]}$. $E_2$ provides the upper bounds for collision edges at different levels. Note that a subgraph with diameter $\leq m^-$ is a tree with high probability since there is no collision edges.

\item {\bf Detailed collision edges in level $> D + t$.}  Define $E_3$ to be the event that
\[
\forall v \in \singleLevelSet_{D + t + 1}, \quad \psi(v)< (1-\delta)^3\mu q \theta.
\]

\item{\bf Infected nodes.}
 Define
\[
E_4 = \{\forall v \in \cup_{i = 0}^{t - 1}{\cal Z}_i, \quad \psi'(v) > (1-\delta)^2\mu q.\}
\]
and
\[
E_5 = \{\forall v \in \cup_{i = 0}^{t - 1}{\cal Z}_i, \quad \psi''(v) > (1-\delta)^3\mu q\theta.\} \]

\item{\bf Infected nodes from source $s_1$.} Define

    \[
        E_6 = \{\tilde{Z}^1_1\geq (1-\delta)^2\mu q\}\cap \{\forall v \in \tilde{Z}_1^1,
        \tilde{Z}'^t_t(v)\geq [(1-\delta)^2\mu q]^{t-1}(1-\delta)\theta\}
    \]
\end{itemize}

To prove these event happens with a high probability, we have
\begin{align*}
    &\Pr(E_1\cap E_2 \cap E_3 \cap E_4\cap E_5\cap E_6) \\
    \geq&\Pr(E_1)(1-\Pr(\bar{E}_2|E_1) -
    \Pr(\bar{E}_3|E_1)  -
    \Pr(\bar{E}_6|E_1)) - \Pr(\bar{E}_4\cup\bar{E}_5)\\
=&\Pr(E_1)(1-\Pr(\bar{E}_2|E_1) -
    \Pr(\bar{E}_3|E_1) -
    \Pr(\bar{E}_6|E_1))- (1- \Pr(E_4\cap E_5))\\
\geq& 1-\epsilon,
\end{align*}
for sufficiently large $n$.
\begin{figure}
        \centering
 		\includegraphics[width=0.7\columnwidth]{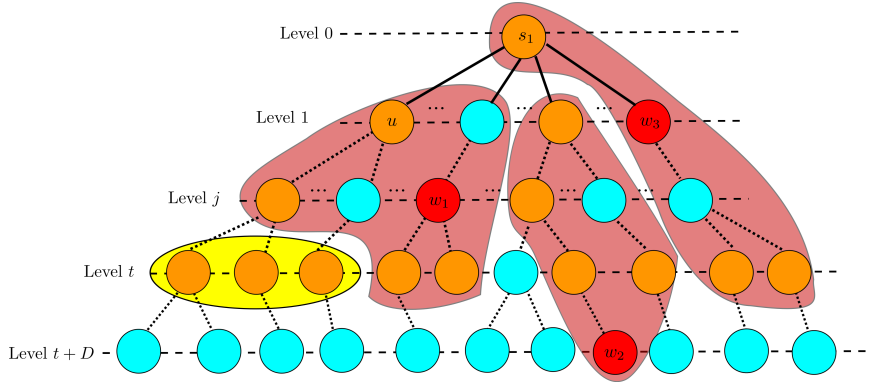}
        \caption{A pictorial example for Theorem 2.}\label{fig:pictorialEg}
\end{figure}
Based on Lemma~\ref{lem:E1E2}, \ref{lem:E3}, \ref{lem:E4E5} and \ref{lem:E6}, we have events $E_1, E_2, E_3, E_4, E_5, E_6$ happen with a high probability as $n$ is large enough. The proofs can be found in the Appendices.

We summarize the important properties of the OJC algorithm based on the above asymptotic events.
\begin{itemize}
    \item[(a)] Let $Y = (1-\delta)^3\mu q\theta$. In Step 1 of the OJC algorithm, we have $\candidateSet
        \subset \levelset_{t+D}$, i.e., all nodes on the subgraph $g^-$ are within level $t+D$
        of the BFS tree rooted at source $s_1$.
        Based on $E_3$, all nodes on level $t+D+1$ have less than $Y$ observed infected neighbors. In
        addition, all nodes one level $l>t+D+1$ have no infected neighbors since all the infected
        nodes are within level $t+D$.
    \item[(b)] Based on event $E_5$, all nodes which are infected at or before time $t-1$
        have at least $Y$ observed infected
    neighbors. Therefore, we have $\cup_{i = 0}^{t - 1}{\cal Z}_i \subset \candidateSet$
    which implies ${\cal I}\subset\candidateSet^+$ and $g^-$ is a connected graph.
\end{itemize}
Therefore, $g^-$ is a connected graph which contains all the infected nodes and is restricted up to level $t+D$ based on $E_3$ and $E_5$.

Next we will show that $s_1\in{\cal W}^*$ by contradiction. Note, we have $e({\cal S},{\cal I}')=t$. Therefore, $e({\cal W}^*,{\cal I}')\leq t$. We will show that $e({\cal W},{\cal I}')> t$
for any set of nodes ${\cal W}$ where $|{\cal W}|= m$, $s_1\not\in{\cal W}.$ Note, all the infection
eccentricity are considered based on the subgraph $g^-$. Specifically, we will show that no nodes in
$g^-$ other than source $s_1$ can reach all nodes which are infected by source $s_1$ within $t$ time
slot based on events $E_1,E_2,E_3,E_4,E_5$ and $E_6$.

Consider the BFS tree rooted at source $s_1$. We discuss the proof in two cases.

\begin{itemize}
    \item [(a)] When $t+D\leq\lfloor m^-\rfloor$, according to the event $E_2$, the $t+D$ hops from $s_1$ is a tree
because there is not collision edges. When event $E_6$ occur, there are at least
$(1-\delta)^2\mu q$ infected nodes at level $1$ and $\forall v \in {\cal Z}^1_1, Z''^t_t(v)\geq
[(1-\delta)^2\mu q]^{t-1}(1-\delta)\theta$ which means there exists at least one observed infected node on
level $t$ for each subtree rooted at level $1$. Consider a set ${\cal W}$ where $s_1\not\in{\cal
W}$. There exists a node $u \in \tilde{\cal Z}^1_1$
such that ${\cal W} \cap \nodes(T^{-s_1}_u)\neq \emptyset$ where $T^{-s_1}_u$ is the tree rooted at
node $u$ without the branch of node $s_1$ (the subtree on level $1$). Therefore, for one node $w\in
\tilde{\cal Z}'^t_t(u)$, we have $d(w,{\cal W})> t$. Hence for any set ${\cal W}$ with size
$m$ that does not contain $s_1,$ $e({\cal W},{\cal I}')>t$. Therefore, we have $s_1\in{\cal W}^*$
when $t+D\leq\lfloor m^-\rfloor$.

Figure \ref{fig:pictorialEg} shows a pictorial example when $m = 3.$ The red nodes are the nodes in set ${\cal W}.$ The existence of node $u$ is guaranteed since $m$ nodes are insufficient to cover all level 1 branches of the BFS tree rooted in $s_1.$ The red areas are the $t$ hop neighborhood of nodes $w_1,w_2$ and $w_3.$ In this case, the $t+D$ neighborhood of node $s_1$ is a tree. Therefore, any set ${\cal W}$ does not contain $s_1$ can not reach the yellow area $\tilde{\cal Z}''^t_t(u)$ within $t$ hops as shown in the figure.

\item[(b)] Consider the case when $t+D>\lfloor m^- \rfloor$. Again, based on event $E_6$, there exists a node $u \in \tilde{\cal Z}^1_1$ such that ${\cal W} \cap \nodes(T^{-s_1}_u)\neq \emptyset$.
The distance between a node in ${\cal Z}'^t_t(u)$ and set ${\cal
W}$ on the BFS tree is larger than $t$. Therefore, if ${\cal W}$ is a Jordan infection cover, the
shortest paths between ${\cal Z}'^t_t(u)$ and set ${\cal W}$ must contain at least one collision
nodes.  In the rest of the proof, we will show that the number of collision nodes is insufficient to
provides the ``shortcuts'' to all observed infected nodes.
\end{itemize}

Define $H$ to be the total number of nodes each of which has the shortest path to ${\cal W}$ within
$t$ hops and containing at least one collision node.  If $H<\tilde{Z}'^\obsTime_\obsTime(u),$ there exists a node $w\in{\cal Z}^\obsTime_\obsTime(u)$ such that
$\dist(w,{\cal W})>t$. Therefore, ${\cal W}$ can not be the Jordan infection cover and the theorem is proved.

In the rest of the proof, we will show that $H<\tilde{Z}'^\obsTime_\obsTime(u)$. We first have the lower bound on $\tilde{Z}'^\obsTime_\obsTime(u)$ according to $E_6,$
\begin{align}
    \tilde{Z}'^\obsTime_\obsTime(u)\geq [(1-\delta)^2\mu q]^{t-1}(1-\delta)\theta\label{eqn:ztt_upperbound}
\end{align}
For each node $w_i$ in ${\cal W},$ define $H_i$ to be the total number of nodes each of which has the shortest
path to $w_i$ within $t$ hops and containing at least one collision node.
The upper bound of $H_i$ can be obtained based on Lemma 6  in \cite{ZhuYin_15_arxiv}.
\[
H_i\leq
c[(1+\delta)\mu]^{\frac{3}{4}(t+D)+\frac{1}{2}}+c[(1+\delta)\mu]^{(\frac{5}{4}-\frac{\alpha}{2})(t+D)+2},
\]
and we have
\[
    H \leq \sum_{i = 2}^m H_i \leq mc[(1+\delta)\mu]^{\frac{3}{4}(t+D)+\frac{1}{2}}+mc[(1+\delta)\mu]^{(\frac{5}{4}-\frac{\alpha}{2})(t+D)+2},
\]

 Since $\frac{1}{2}<\alpha<1,$ we have $\alpha=\frac{1}{2}+\alpha'$ where $0<\alpha'<\frac{1}{2}$ is a constant,
we have
\begin{align}
\frac{H}{\tilde{Z}'^\obsTime_\obsTime(u)} &\leq \frac{mc[(1+\delta)\mu]^{\frac{3}{4}(t+D)+\frac{1}{2}}}{[(1-\delta)^2\mu
q]^{t-1}(1-\delta)\theta}+\frac{mc[(1+\delta)\mu]^{(\frac{5}{4}-\frac{\alpha}{2})(t+D)+2}}{[(1-\delta)^2\mu
q]^{t-1}(1-\delta)\theta}\\
&=
\frac{mc}{\mu}\left(\frac{(1+\delta)^{\frac{3}{4}+\left(\frac{3D}{4}+\frac{1}{2}\right)\frac{1}{t}}}{[(1-\delta)^2q]^{1-\frac{1}{t}}(1-\delta)\theta\mu^{\frac{1}{4}-\left(\frac{3D}{4}+\frac{5}{2}\right)\frac{1}{t}}}\right)^{t}\\
&+\frac{mc}{\mu}\left(\frac{(1+\delta)^{1-\frac{\alpha'}{2}+\left[\left(1-\frac{\alpha'}{2}\right)D+2\right]\frac{1}{t}}}{[(1-\delta)^2q]^{1-\frac{1}{t}}(1-\delta)\theta\mu^{\frac{\alpha'}{2}-\left(4+\left(1-\frac{\alpha'}{2}\right)D\right)\frac{1}{t}}}\right)^{t}\label{eqn:HvsZ}
\end{align}
Since $t>>D,$ we have
\[
    t\geq \min\left\{3D + 2,\left(\frac{2}{\alpha'}-1\right)D +
    \frac{4}{\alpha'},\left(\frac{4}{\alpha'}-2\right)D + \frac{16}{\alpha'}  \right\}
\]
Therefore, Inequality (\ref{eqn:HvsZ}) becomes
\[
\frac{H}{\tilde{Z}'^\obsTime_\obsTime(u)}\leq
\frac{2mc}{\mu}\left(\frac{(1+\delta)}{[(1-\delta)^2q]\mu^{\frac{\alpha'}{4}}}\right)^{t}.
\]
Since $\mu>\frac{1}{Cq\theta}\log n$ and $\delta,q,\alpha'$ are constants, we have
\[
\frac{(1+\delta)}{[(1-\delta)^2q]\mu^{\frac{\alpha'}{4}}}<1
\]
when
\[
n> \exp\left(\frac{1}{2}\left(\frac{(1+\delta)}{(1-\delta)^2q}\right)^{\frac{4}{\alpha'}}\right).
\]
Therefore, we have
\[
\frac{H}{\tilde{Z}'^\obsTime_\obsTime(u)}\leq \frac{2mc}{\mu}\leq \epsilon',
\]
where $\epsilon'\in(0,1)$ is a constant and the inequality holds for sufficiently large $n$.
There are at least $(1-\epsilon')\tilde{Z}'^\obsTime_\obsTime(u)$ nodes which cannot be
reached from ${\cal W}$ with $t$ hops on $g^-$. Hence we have $\ecce({\cal I}',{\cal W})>t$.
Therefore, we proved that $s_1\in{\cal W}^*$ with a high probability when $n$ is sufficiently large, i.e., we have
\[
    \Pr(s_1\in{\cal W}^*)\geq 1-\frac{\epsilon}{m}
\]
since $m$ is a constant. Then, by applying the union bound, we have, for sufficiently large $n,$
\[
    \Pr(s_i\in{\cal W}^*,\quad\forall i = 1,\cdots,m)\geq 1 - \epsilon
\]
Note, we have $|{\cal W}^*|= m$. Therefore, we have
\[
    \Pr({\cal S} = {\cal W}^*)\geq 1 - \epsilon
\]
Hence, the Jordan infection cover equals the actual source set with a high probability.
\end{proof}

We now briefly explain the conditions.  Recall that $\mu$ is the average node degree, $q$ is the lower bound on the infection probability and $\theta$ is the lower bound on the reporting probability, so $\mu q \theta$ is a lower bound on the average number of observed infected neighbors of a node that was infected before time slot $t.$ Therefore, condition (c1) requires that this lower bound is $\Omega(\log n),$ and condition (c2) requires that the threshold used in the candidate selection algorithm is a constant fraction of the average number of observed infected neighors. Applying the Chernoff bound, conditions (c1) and (c2) together yield the following conclusions:
\begin{itemize}
\item[(i)]  any node who was infected before or at time slot $t-1$ (hence, including the sources) will be selected into the candidate set with a high probability,

\item[(ii)] any node that is $t + D + 1$ hops away from the set of sources will not have $Y$ or more observed infected neighbors with a high probability, and

\item[(iii)] any node that is more than $t+D+1$ hops away from the set of sources will not have any observed infected neighbors.
\end{itemize}
Based the above facts, with a high probability, the candidate set includes all nodes who were infected at $t - 1$ or earlier, and any node in $g^-$ is at most $t+D$ hops away from all sources.

Condition (c3) is on the infection duration. We first restrict $\obsTime=\omega(D)$ so that the infection subgraphs starting from different sources are likely to overlap and form a connected component. This is a more interesting regime than the one in which infection subgraphs are disconnected from each other.  $\lim\sup_{n\rightarrow \infty} \frac{\obsTime}{ \frac{\log n}{\log\mu}}<\frac{2}{3}$  is a critical condition. The intuition why it is required is explained below.  Figure \ref{fig:pictorialEg} provides a pictorial explanation of the proof. The picture illustrates the breadth-first-search (BFS) tree $\tree^\dag$ rooted at source $s_1$ with height $t+D,$ where $s_1$ is one of the $m$ sources. The nodes in orange are the observed infected nodes whose infection was originated from $s_1.$ The blue nodes are unobserved nodes. A node is said to be on level $i$ of the BFS if its hop-distance to $s_1$ is $i.$ Assume $m=3$ and consider a set of three nodes who are within $t+D$ hops from $s_1$ but not includes $s_1$ (e.g., ${\cal W} = \{w_1, w_2, w_3\}$ in Figure \ref{fig:pictorialEg}). Suppose the infection eccentricity of ${\cal W}$ is $\leq t.$ Since $s_1$ has a sufficient number of neighbors according to the definition of $\mu,$ with a high probability, there exists a subtree of $\tree^\dag$ starting from an offspring of $s_1,$ which does not include any node in ${\cal W}$. Assume $u$ is the root of such a subtree in Figure \ref{fig:pictorialEg}. The yellow area in Figure \ref{fig:pictorialEg} includes the level-$t$ observed infected nodes on subtree $T_u^{-s_1}.$ Any path from $w_1,$ $w_2$ or $w_3$ to the yellow area, formed by edges on $\tree^\dag,$ must have hop-distance larger than $t.$ Therefore, if the infection eccentricity of ${\cal W}$ is at least $t,$ there must exist a path from ${\cal W}$ to each of the nodes in the yellow area with hop-distance $\leq t,$ and such a path must include at least one edge which is not in $\tree^\dag$ (we call these edges \emph{collision edges}). In the detailed analysis, we will prove that with a high probability, \textcolor{black}{the number of nodes within $t$ hops from $\cal W$ via the collision edges} is order-wise smaller than the number of nodes in the yellow area when (c3) holds. Therefore, the Jordan cover has to include $s_1$. The same argument applies to other sources.

\subsection{Impossibility results}\label{sec:ImpossibilityResult}
Theorem 5 in \cite{ZhuYin_16} presents the conditions under which it is impossible to identify the single source under the IC diffusion on the ER random graph, which is a special case of the model in this paper.  Assuming SI or IC model, based on Lemma 1 in \cite{ZhuYin_16}, we have that with a high probability, all nodes of the ER graph become infected when the infection duration is at least
\[
t_u \triangleq \left\lceil\frac{\log n}{\log \mu + \log q}\right\rceil +2.
\]
When this occurs, it is impossible to detect the sources since the nodes are indifferentiable.
\begin{thm}\label{thm:impossibility}
Assume the multi-source diffusion follows the IC or SI model. If $24\log n<q\mu <<\sqrt{n}$ and $q$ is a constant independent of $n,$ then
	\[
	\lim_{n\rightarrow \infty}\Pr({\cal I}=\nodes(g))=1
	\]
	when the observation time $t\geq t_u.$	In other words, the entire network is infected after $t_u$ with a high probability. In such a case, the probability of finding the sources diminishes to zero as $n\rightarrow \infty.$
\end{thm}

\section{Simulations}\label{sec:Simulation}
\begin{figure*}
\centering
\begin{subfigure}[b]{0.25\textwidth}
\includegraphics[width=\textwidth]{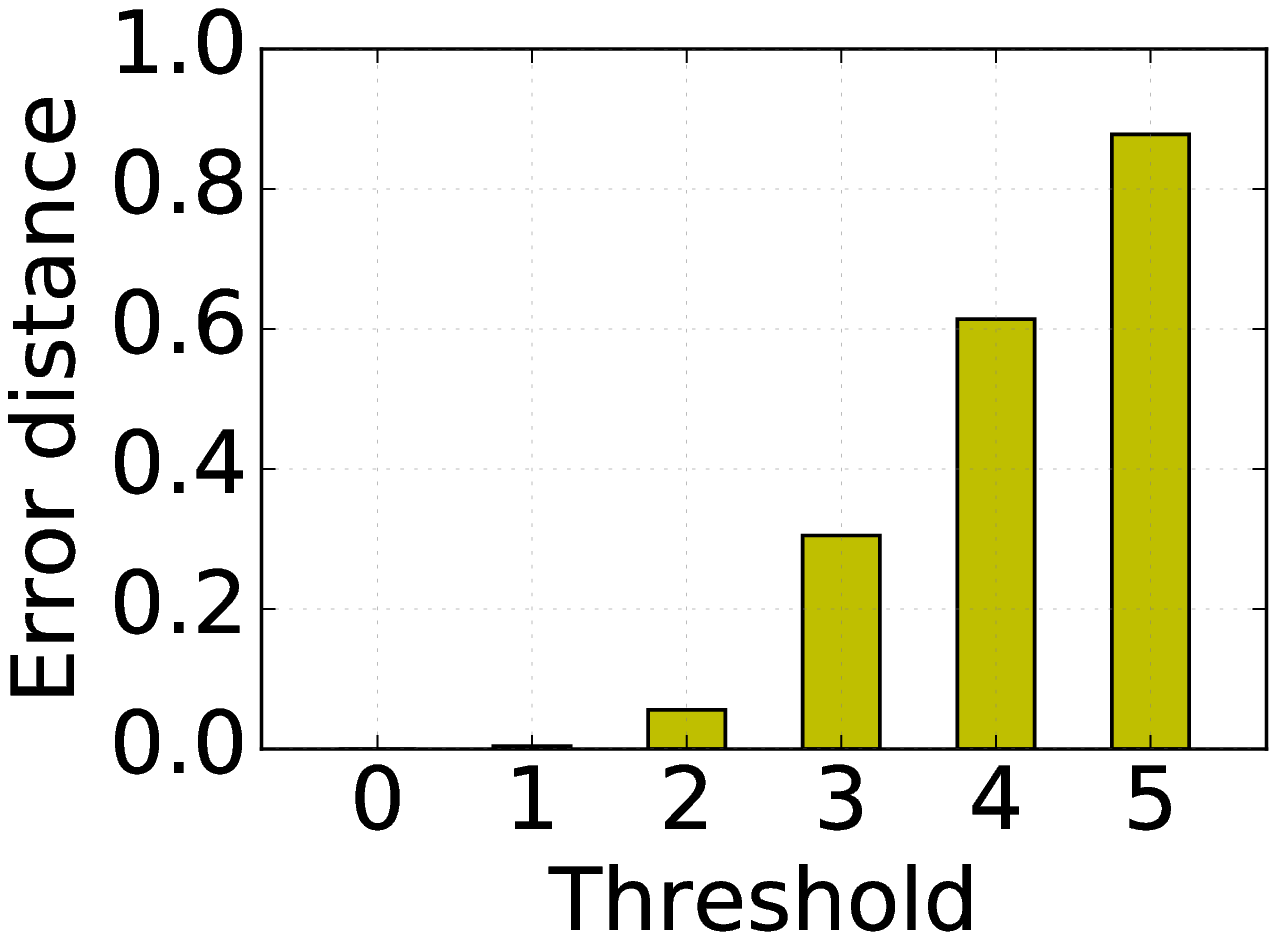}
\caption{Error distance}
\label{fig:er-o-ed}
\end{subfigure}
\hspace{1em}
\begin{subfigure}[b]{0.25\textwidth}
\includegraphics[width=\textwidth]{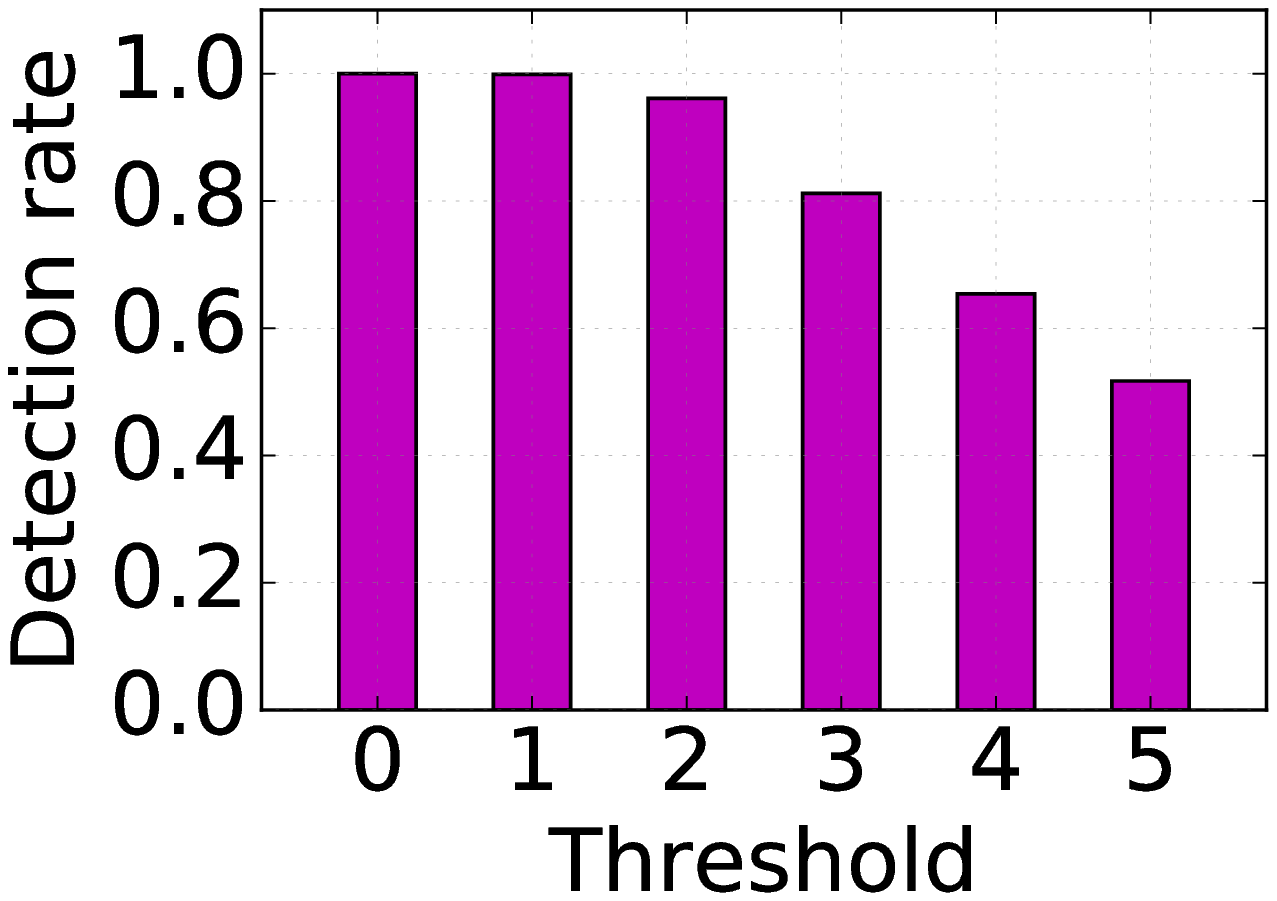}
\caption{Detection rate}
\label{fig:er-o-dr}
\end{subfigure}
\hspace{1em}
\begin{subfigure}[b]{0.25\textwidth}
\includegraphics[width=\textwidth]{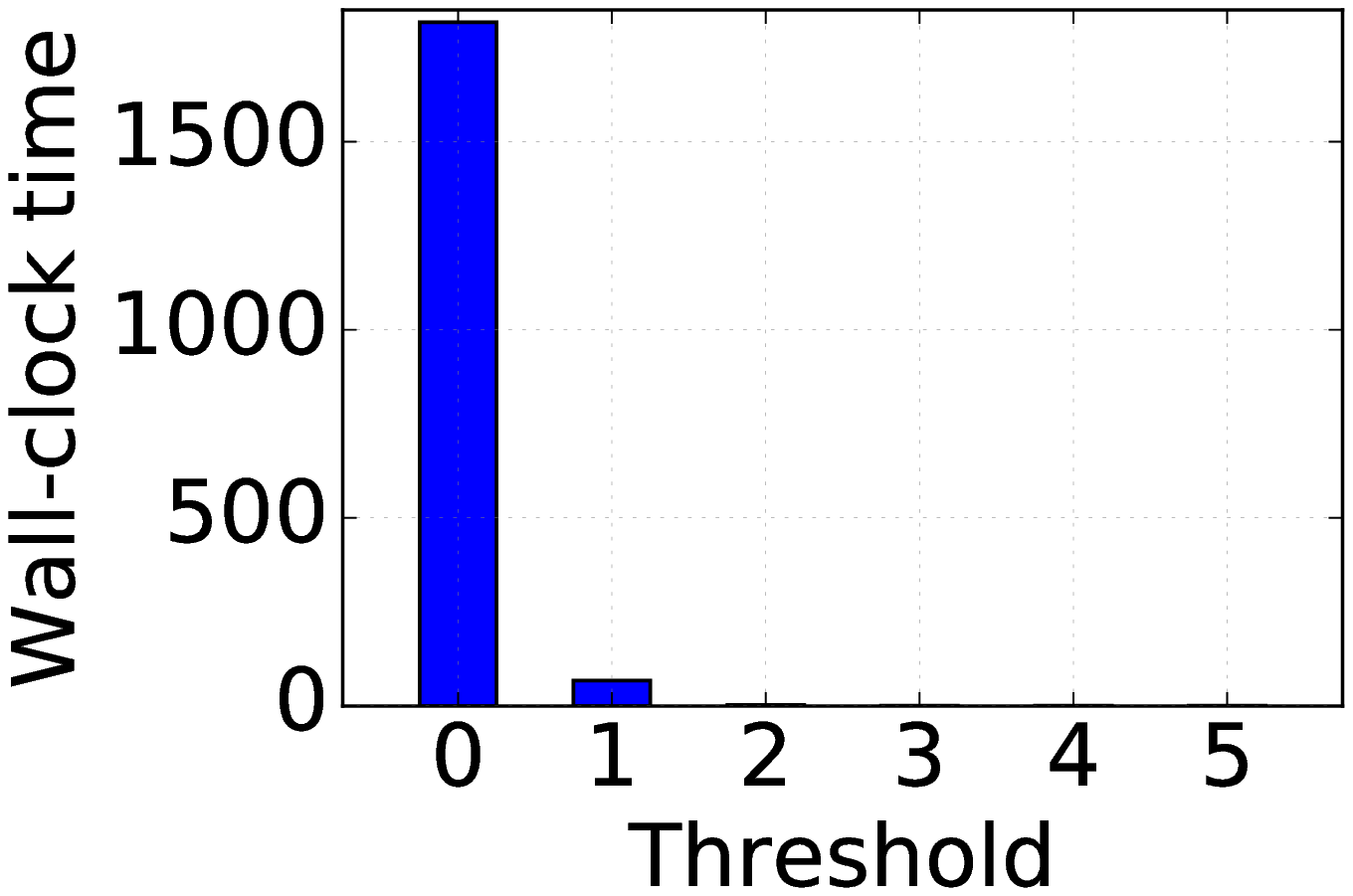}
\caption{Wall-clock time}
\label{fig:er-o-t}
\end{subfigure}
\caption{Performance of OJC with different threshold values on the ER random graph}
\label{fig:er_optimal}
\end{figure*}
\begin{figure*}
\centering
\begin{subfigure}[b]{0.25\textwidth}
\includegraphics[width=\textwidth]{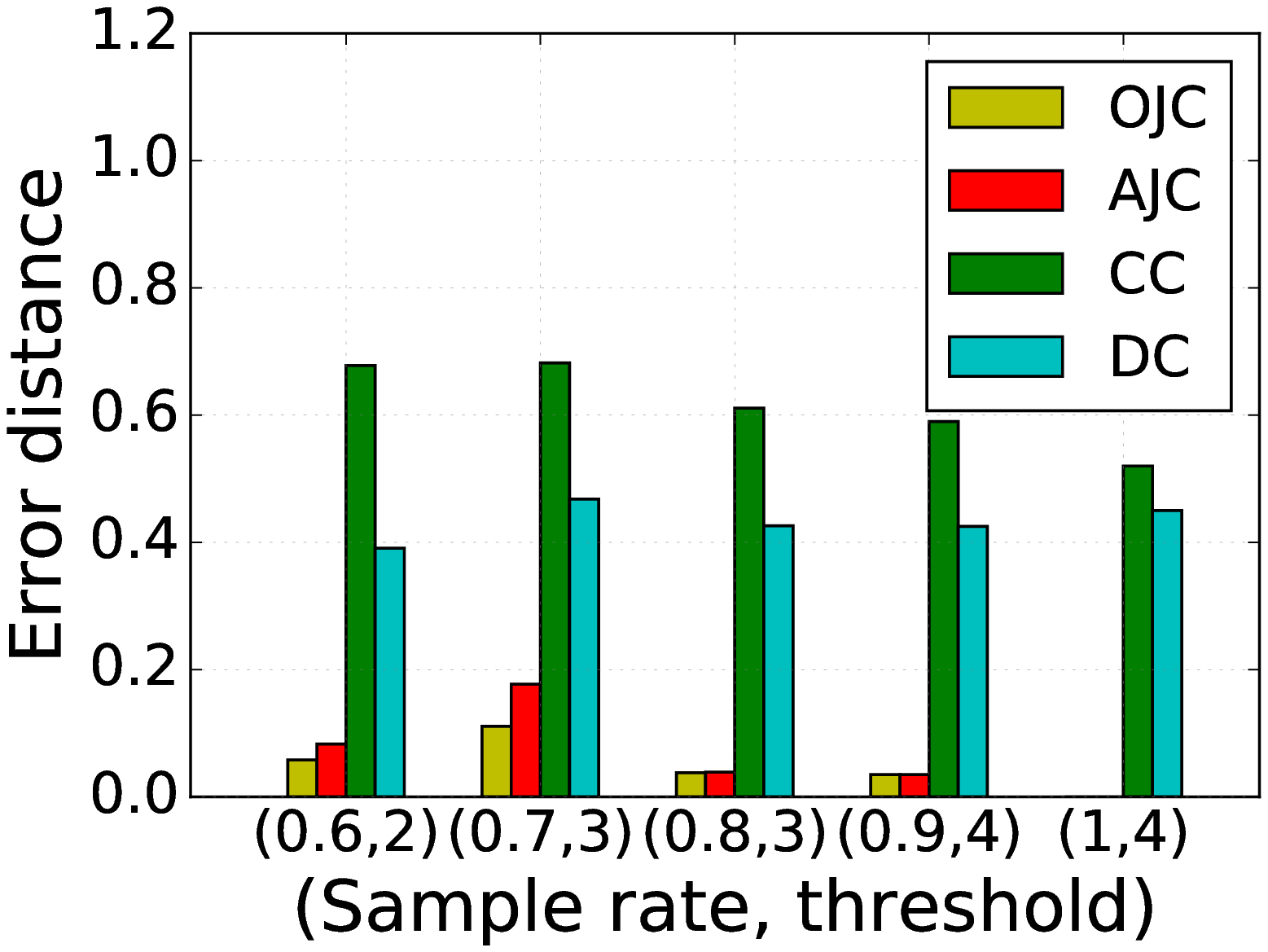}
\includegraphics[width=\textwidth]{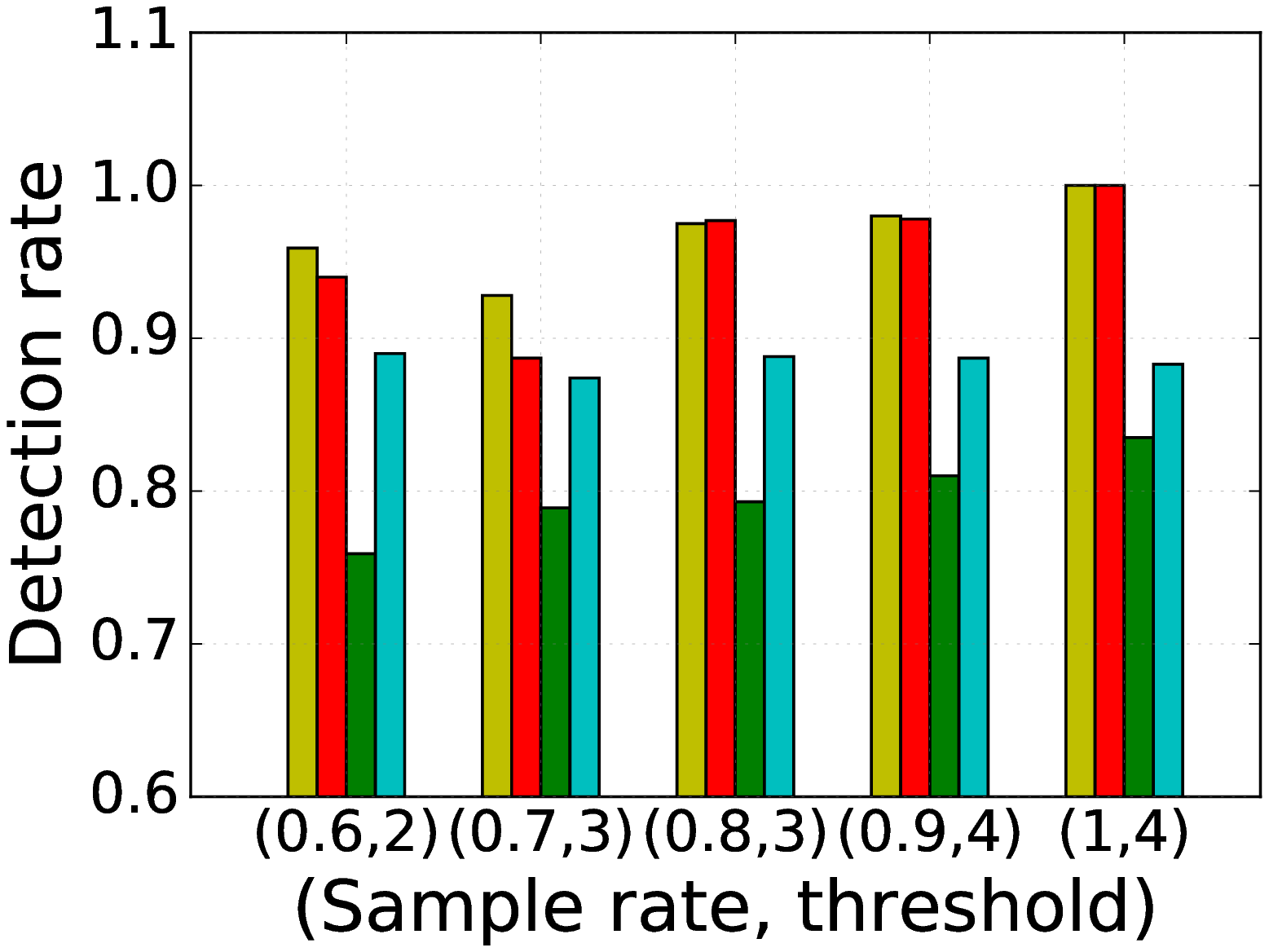}
\caption{Source number: 2, infection size: $100\sim 300.$}
\label{fig:er2}
\end{subfigure}
\hspace{1em}
\begin{subfigure}[b]{0.25\textwidth}
\includegraphics[width=\textwidth]{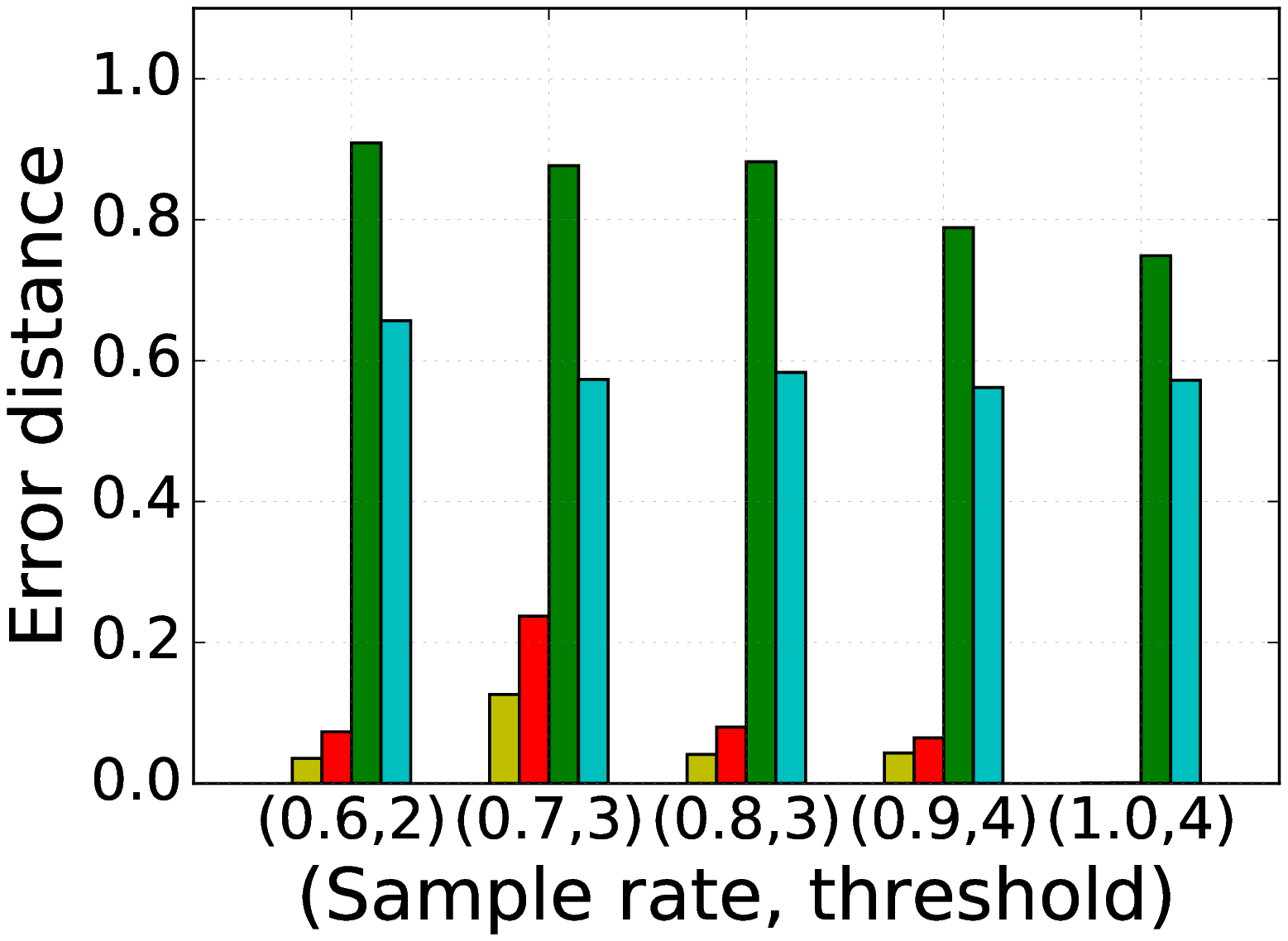}
\includegraphics[width=\textwidth]{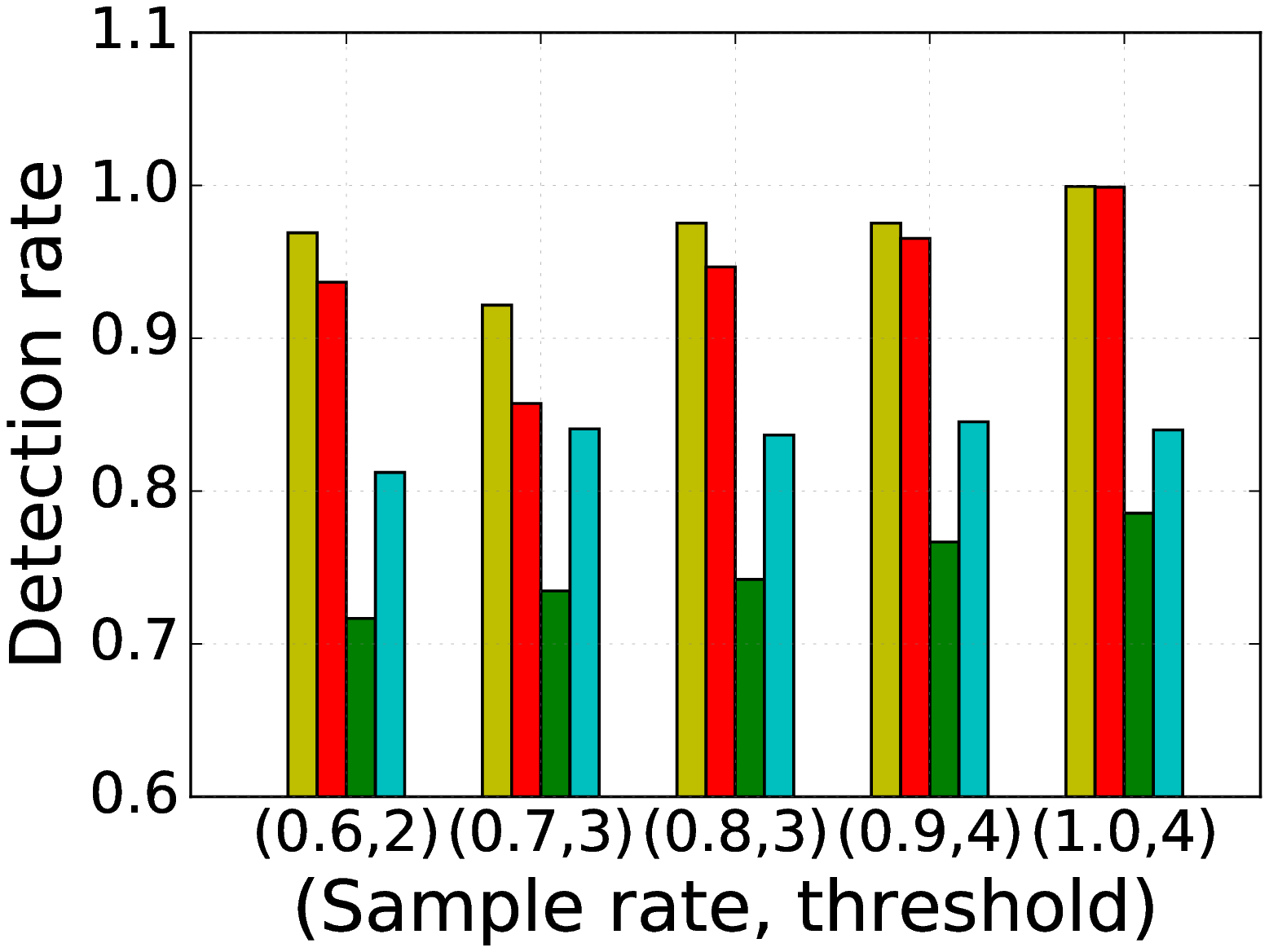}
\caption{Source number: 3, infection size: $200\sim 400.$}
\label{fig:er3}
\end{subfigure}
\hspace{1em}
\begin{subfigure}[b]{0.25\textwidth}
\includegraphics[width=\textwidth]{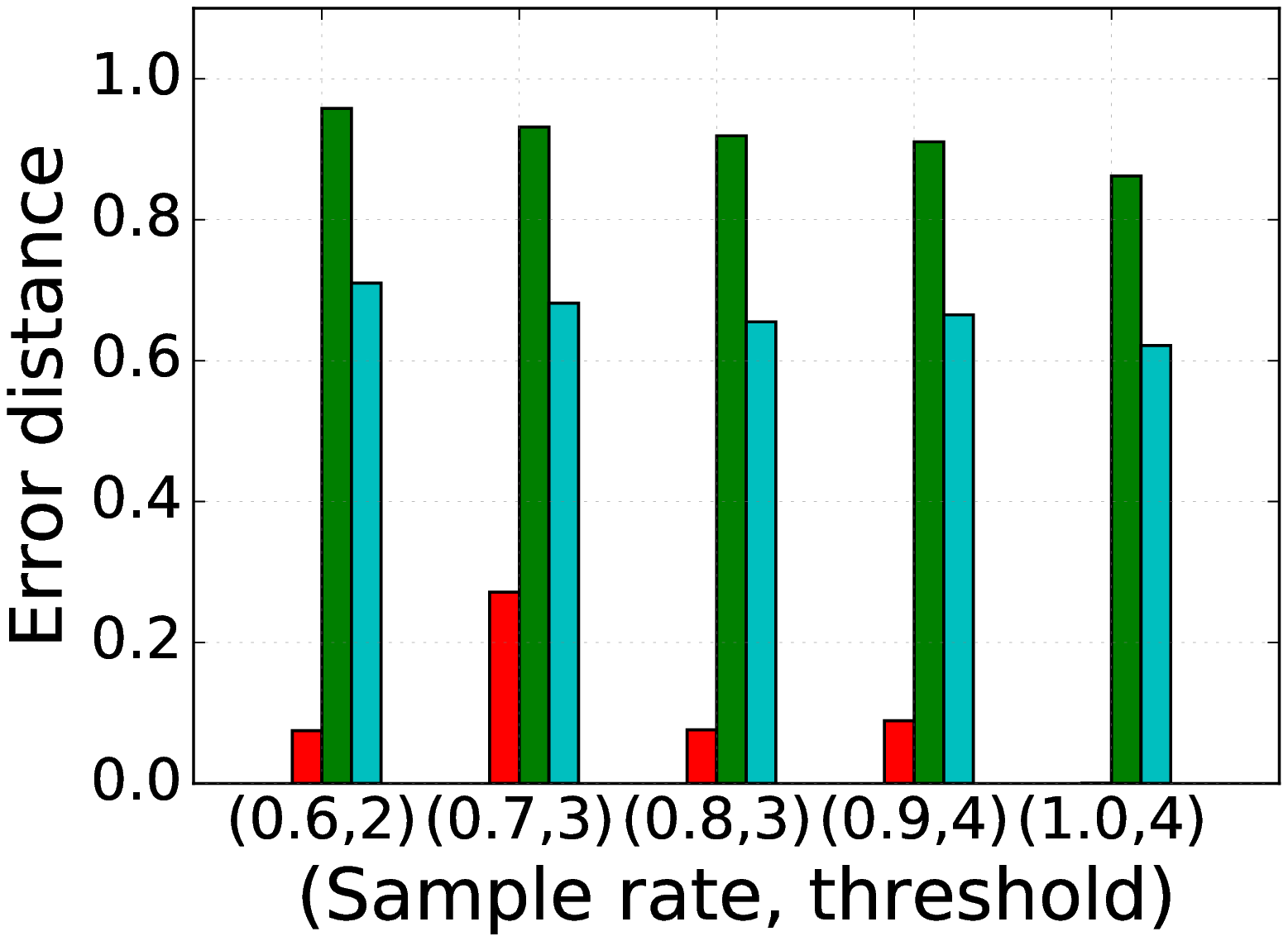}
\includegraphics[width=\textwidth]{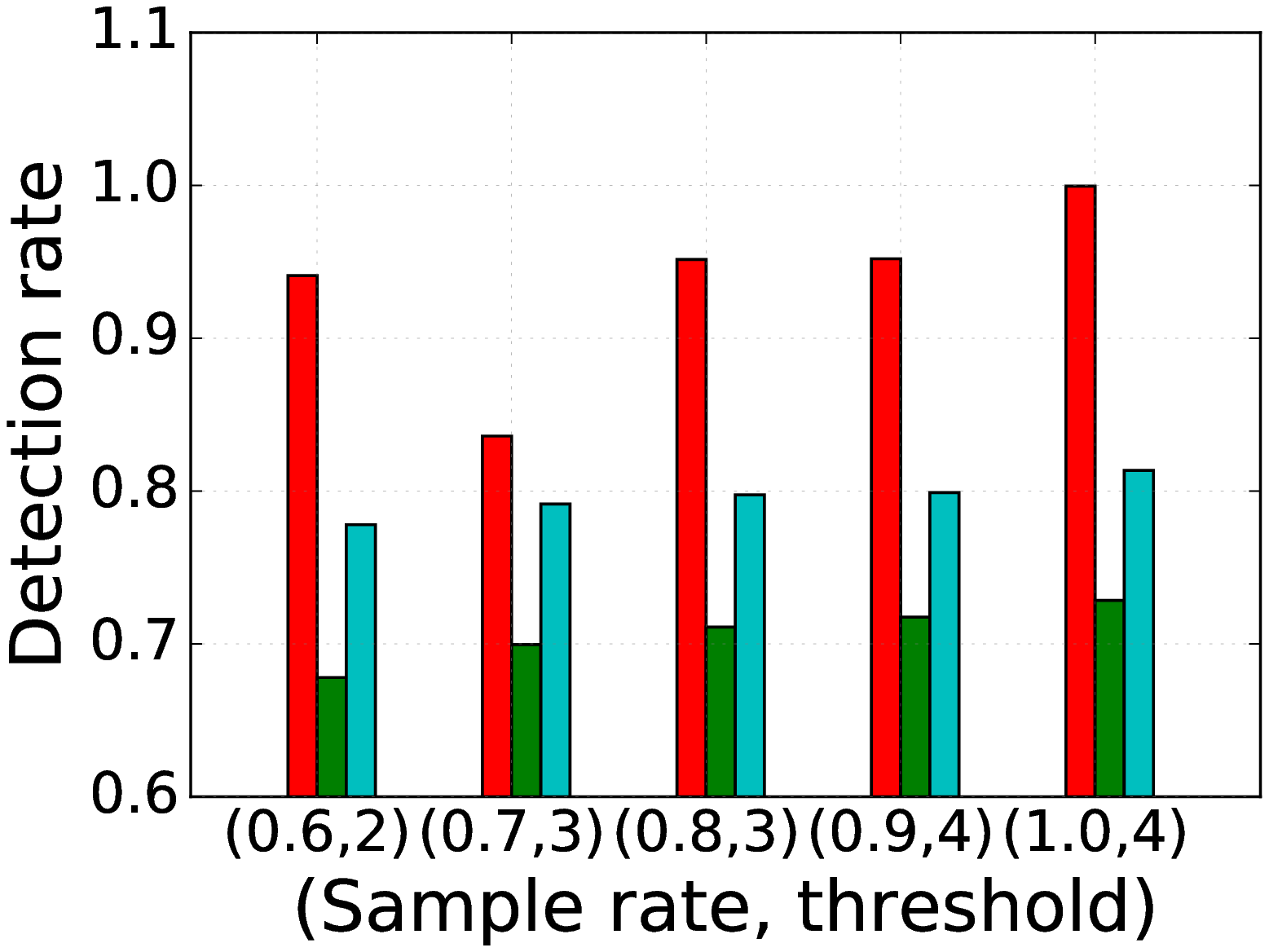}
\caption{Source number: 4, infection size: $300\sim 500.$}
\label{fig:er4}
\end{subfigure}
\caption{The Performance of OJC, AJC, CC and DC on the ER random graph with different sample rates and threshold values}
\label{fig:er}
\end{figure*}

\begin{figure*}
\centering
\begin{subfigure}[b]{0.25\textwidth}
\includegraphics[width=\textwidth]{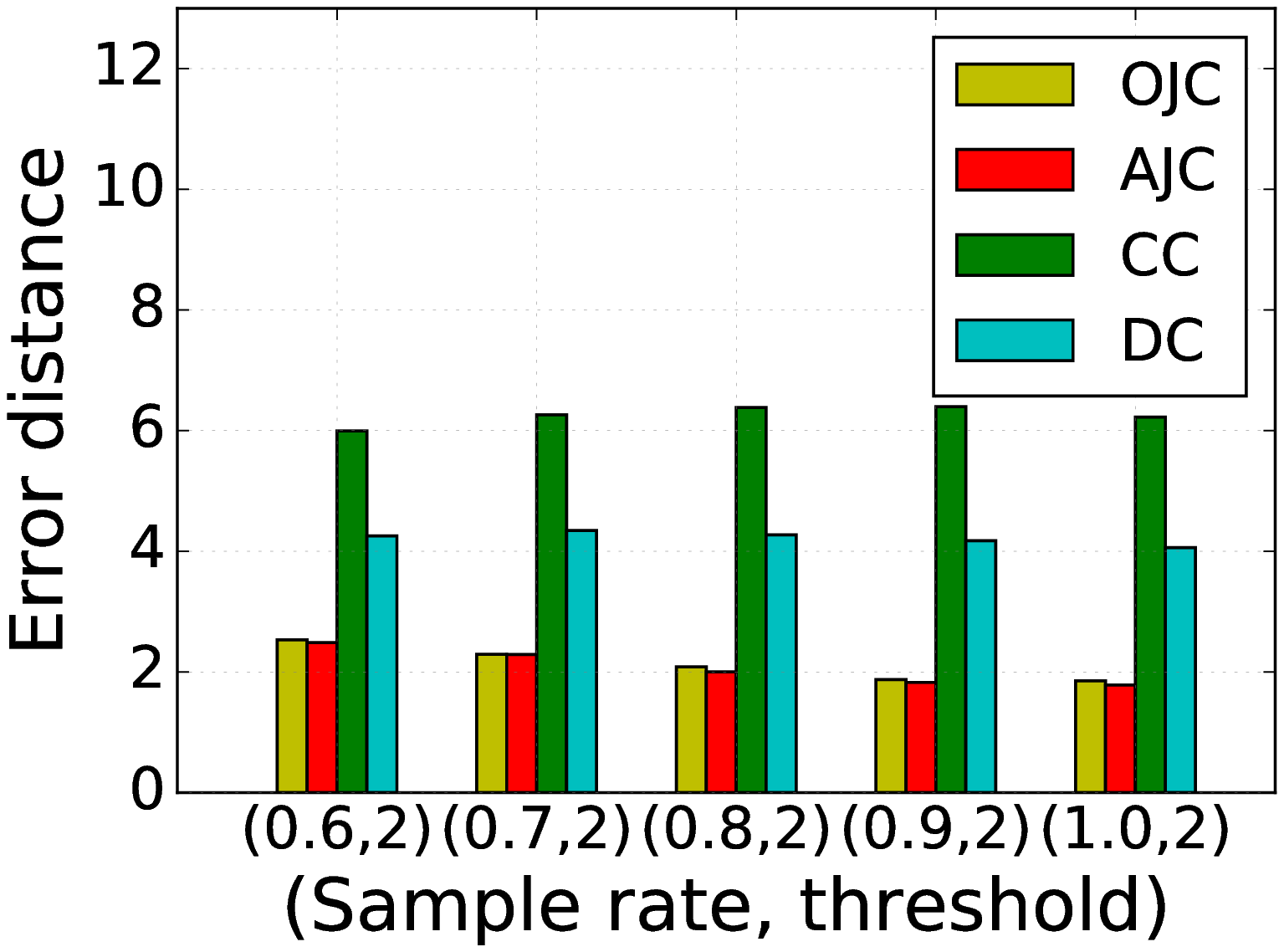}
\includegraphics[width=\textwidth]{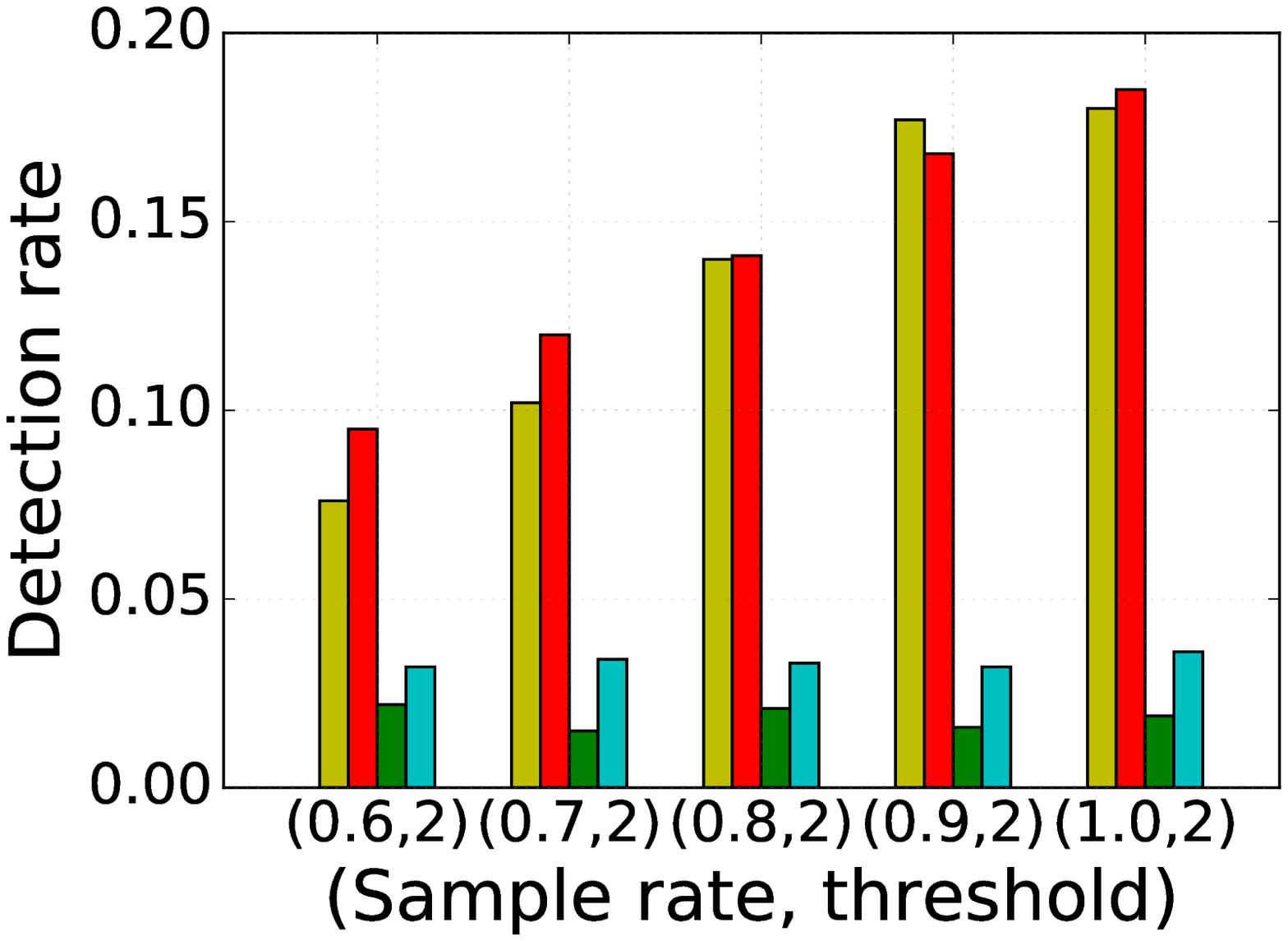}
\caption{Source number: 2, infection size: $100\sim 300.$}
\label{fig:pg2p1}
\end{subfigure}
\hspace{1em}
\begin{subfigure}[b]{0.25\textwidth}
\includegraphics[width=\textwidth]{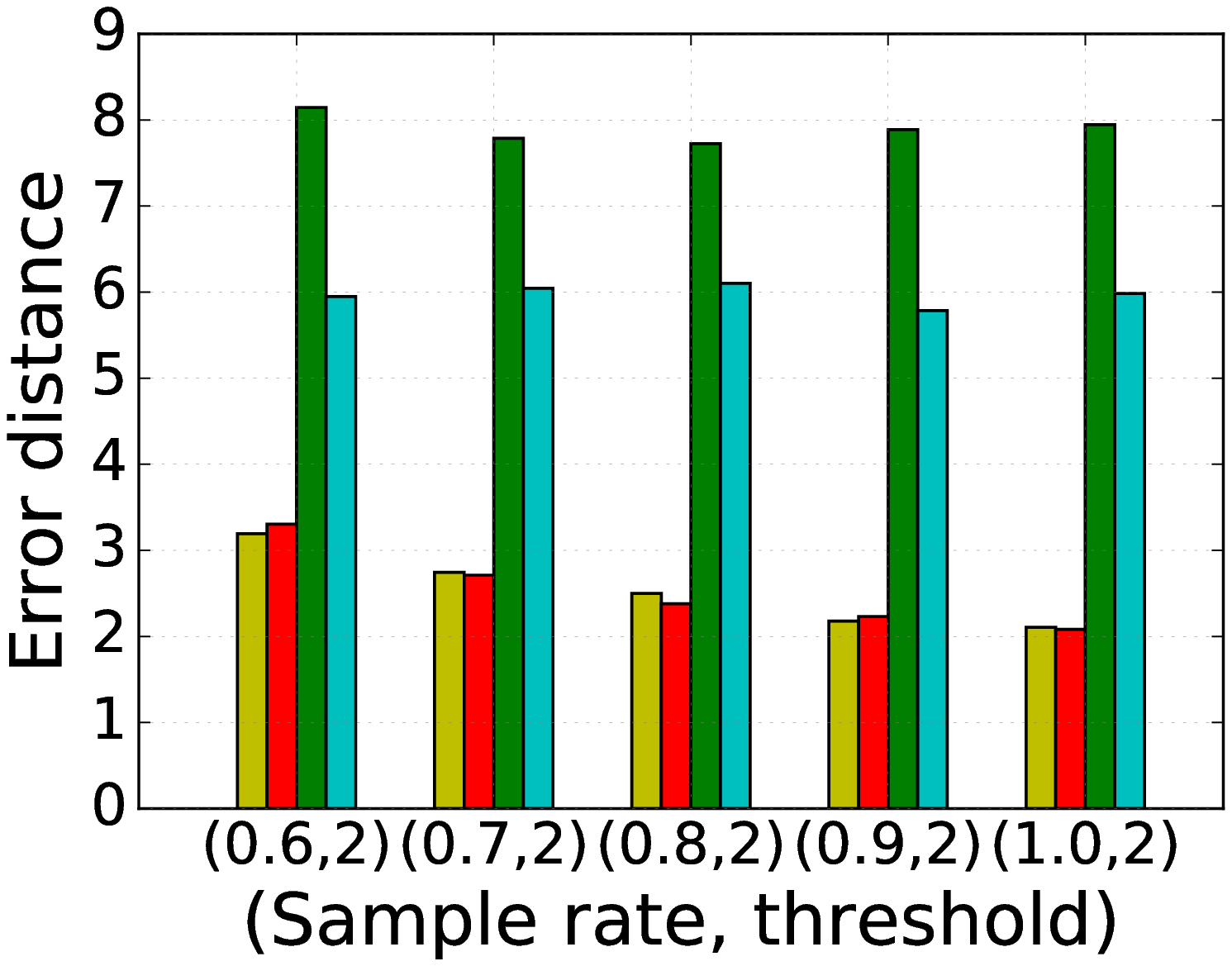}
\includegraphics[width=\textwidth]{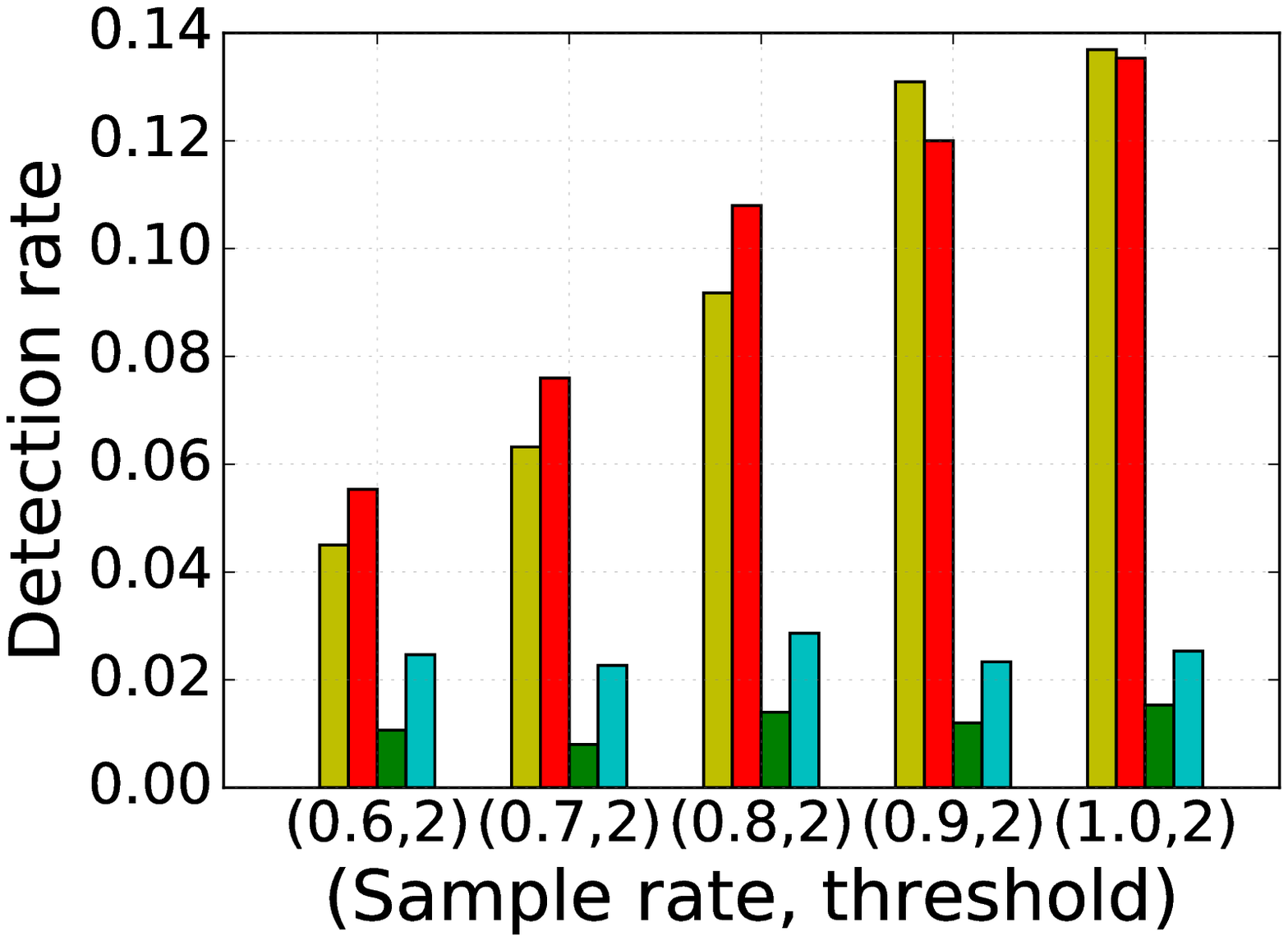}
\caption{Source number: 3, infection size: $200\sim 400.$}
\label{fig:pg3p1}
\end{subfigure}
\hspace{1em}
\begin{subfigure}[b]{0.25\textwidth}
\includegraphics[width=\textwidth]{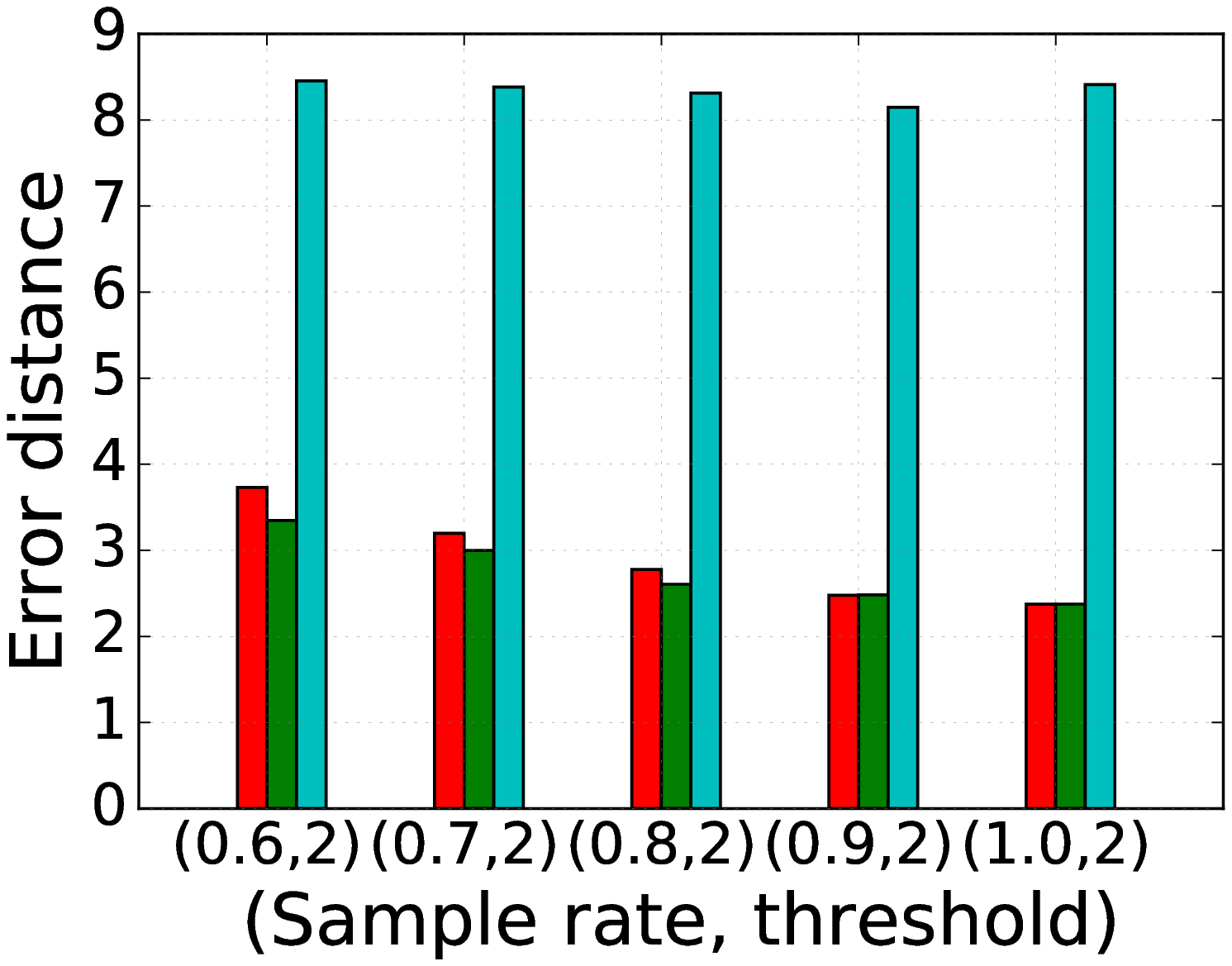}
\includegraphics[width=\textwidth]{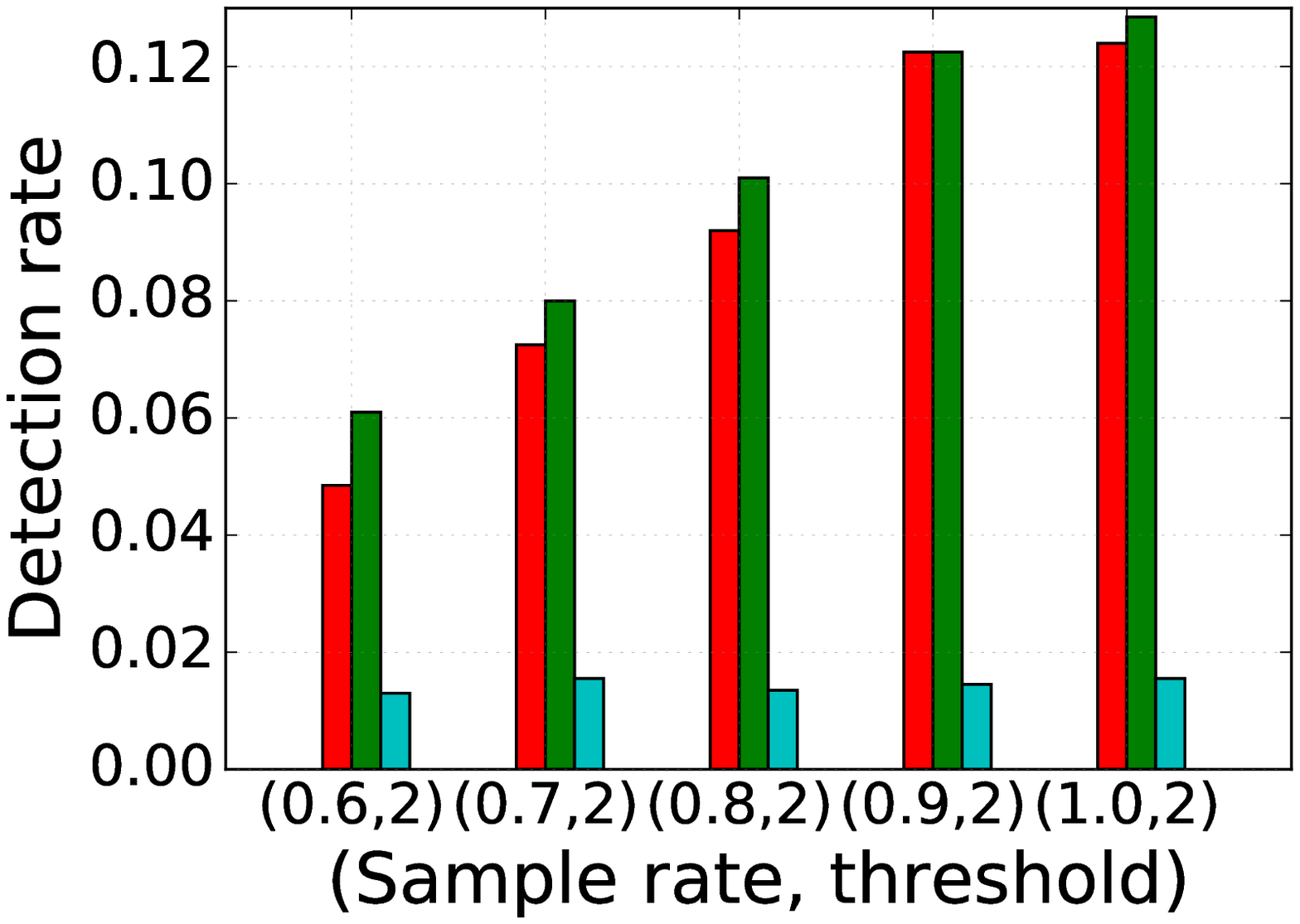}
\caption{Source number: 4, infection size: $300\sim 500.$}
\label{fig:pg4p1}
\end{subfigure}
\caption{The Performance of OJC, AJC, CC and DC on the power grid network with different sample rates and threshold values}
\label{fig:pg}
\end{figure*}

In this section, we evaluated the performance of our algorithms via simulations. The performance metrics used in this paper include:
\begin{itemize}
\item Error distance: The error distance is defined to be
\begin{equation*}
\min_{{\cal P}\in \text{permutation}(\estimatorSet)} \sum_{i=1}^m \frac{d(s_i,p_i)}{m},
\end{equation*}
where $s_1,s_2,\dots,s_m$ are the real sources, $\estimatorSet$ is the set of detected sources and ${\cal P}=(p_1,p_2,...,p_m)$ is a permutation of $\estimatorSet.$
\item Detection rate: The detection rate is defined as $$\frac{|\sourceSet \cap \estimatorSet|}{m}.$$
\end{itemize}
We compared our algorithms with two heuristic algorithms (DC and CC) based on K-Means, which have been used for comparison in \cite{LuoTayLen_14}. The algorithms proposed in \cite{LuoTayLen_14} and \cite{CheZhuYin_14} are the same as AJC without candidate selection. In both DC and CC, the initial centroids are randomly chosen. During the clustering step of each iteration in K-Means, we selected distance centroid of each cluster in DC and selected closeness centroid in CC, where distance centroid is defined as $\arg\min_{v\in \cal C} \sum_{u\in \cal C\cap \infObs} d(v,u),$ where $\cal C$ is the set of nodes in the cluster. Closeness centroid is defined as $\arg\max_{v\in \cal C} \sum_{u\in {\cal C}\cap \infObs, u\neq v}\frac{1}{d(u,v)}.$ The following experiments were conducted on an server with 8 Intel Xeon X3450 CPUs and 16G RAM with Linux 64 bit system. All algorithms were implemented with Python 2.7.

{
\color{black}
\subsection{OJC with different thresholds}
In Figure \ref{fig:er_optimal}, we evaluated OJC on the ER random graph. In the experiments, we generated an ER random graph with  $5,000$ nodes and wiring probability $0.002.$    We used the homogeneous SI model for diffusion with infection probability $0.8.$  In this experiment, we limited the infection network size to be $100\sim 300$ and the number of sources to be 2 due to the computational complexity of the OJC algorithm.  Figure \ref{fig:er_optimal} shows the performance of OJC with different thresholds.  From the results, the detection rate is close to one and the error distance is close to zero under OJC with threshold 0 or 1. However, the running time is 1,817 seconds versus 68 seconds. So the candidate selection algorithm with threshold one results in 27$\times$ reduction of the running time. When the threshold increases 2, the running time reduces to 3 seconds, which is a $600\times$ reduction of the running time. Both the detection rate and the error distance became slightly worse in this case. The detection rate in this case is 0.961  and the error distance is 0.056.

\subsection{OJC, AJC and other heuristics}
We further evaluated the performance of OJC and AJC on both the power grid network \cite{WatStr_98} and  ER random graph (size: 5000, wiring probability: 0.002) and compared them with DC and CC heuristics. We used the homogeneous SI model with infection probability $0.8$ to generate the diffusion sequences. For AJC/CC/DC, for each diffusion sequence, we repeated the algorithm 100 times from different initial conditions and chose the source set with the smallest the smallest-infection-eccentricity/largest-closeness-centrality/smallest-distance-centrality.
In Figure \ref{fig:er} and \ref{fig:pg}, the $x-$axis represents the combinations of sample rate and threshold. On the ER random graph, we increased the threshold as the sample rate increased to control the running time. For the power-grid network, since the average node degree is only 2, we set threshold equal to 2 for experiments for all sample rates. As we can see from the figures that when fixing the threshold, the performance of all algorithms (in terms of both error distance and detection rate) improves as the sample rate increases because we had more information about the diffusion. From Figure \ref{fig:er} and \ref{fig:pg}, we can also see that AJC outperforms DC and CC, and has similar performance with OJC. Note that with four sources, OJC became very slow on both the ER random graph and the power grid network because its complexity increases exponentially in the number of sources. So for the cases with four sources, we only simulated AJC.  
}

\section{Conclusions}
In this paper, we studied the problem of detecting multiple diffusion sources under the heterogeneous SIR model with incomplete observations. We defined a concept called Jordan cover and developed the OJC algorithm based on that. Our theoretical analysis showed that OJC finds the set of sources in the ER random graph with probability one asymptotically under mild conditions. To the best of our knowledge, this is the first theoretic performance guarantee for multiple information sources detection in non-tree networks. Since \textcolor{black}{the computational complexity of OJC is polynomial in $n$ but exponential in $m$}, we proposed a heuristic algorithm --- the AJC algorithm. Our simulation results showed that OJC and AJC algorithms have similar performance and both significantly outperform  existing algorithms. 

\section{Acknowledgments}
This work was supported in part by the U.S. Army Research
Laboratory’s Army Research Office (ARO Grant No. W911NF1310279) and by
the Defense Threat Reduction Agency (DTRA Grant HDTRA1-16-0017).

\bibliographystyle{plain}
\bibliography{inlab-refs}
\newpage
\appendix
\appendixpage

\section{Proof of $E_1$ and $E_2$}
\begin{lemma}\label{lem:E1E2}
    Assume the conditions in Theorem 2 hold, for any $\epsilon > 0$, we have
    \[
        \Pr(E_1)\geq  1-\epsilon,
    \]
    and
    \[
        \Pr(E_2|E_1)\geq 1-\epsilon,
    \]
    for suffciently large $n$.
\end{lemma}
The proof follows directly from the proof of Lemma 3 and 4 in \cite{ZhuYin_15_arxiv}.

\section{Proof of $E_3$}
\begin{lemma}\label{lem:E3}
    Assume the conditions in Theorem 2 hold, for any $\epsilon > 0,$ we have
    \[
        \Pr(E_3|E_1)\geq 1-\epsilon,
    \]
    for suffciently large $n$.
\end{lemma}
\begin{proof} Since all sources are within $D$ hops from $s_1$ and the snapshot is taken at time $t,$ all the
 infected nodes are within $t+D$ hops from $s_1$. To prove the conclusion, we only need to show
 that any node on level $t + D + 1$ does not have more than $(1-\delta)^3\mu q \theta$ neighbors.
 Since all infected nodes are within level $t + D$, instead of considering the
 observed infected neighbors,  we only need to show that any node on level $t + D + 1$ does not
 have more than $(1-\delta)^3\mu q\theta$ neighbors in level $t + D$.Therefore, in this proof, we consider a more
 restrictive event which is only a topological feature of the ER random and does not depend on
 the infection process.

Based on $E_1,$ there are at most $[(1+\delta)\mu]^{t + D}$ nodes in
level $t + D$ and at most $[(1+\delta)\mu]^{t + D + 1}$ nodes in level
$t + D + 1$. For any node $v$ in level $t + D + 1,$ the neighbors of node
$v$ in $t + D$ are either the node which introduce node $v$ into level
$t + D + 1$ (i.e., the parent of $v$ in the BFS tree) or the collision
edges between node $v$ and nodes in level $t + D$. The total number of
possible collision edges depends on the order that the parent of node $v$
is introduced to the BFS tree.

In general, if the parent of node $v$ is the $i$th node, the number of
possible neighbors on level $t + D$ follows
$\hbox{Bi}([(1+\delta)\mu]^{D + t} - i, \mu/n) + 1$. As a summary, for
any node on level $t + D + 1$ the number of neighbors in level $t + D$ is
stochastically upper bounded by $\hbox{Bi}([(1+\delta)\mu]^{D + t},\mu n) + 1$.
Define
\[
    X\triangleq (1-\delta)^3\mu q \theta - 1.
\]
Define $\delta'$ to be
\[
\delta' \triangleq \frac{Xn}{[(1+\delta)\mu]^t\mu} - 1.
\]

Denote by $N_v$ the number of neighbors in level $t + D$ for one node $v$ on level $t + D + 1$.
\begin{align}
    \Pr(N_v\geq X + 1|E_1) &\leq \exp\left(-\frac{\delta'^2 [(1+\delta)\mu]^{t+D}\mu/n}{2+\delta'}\right)\\
    &=\exp\left(-\frac{\delta'}{2+\delta'}\times \delta'[(1+\delta)\mu]^{t+D}\mu/n\right)\\
    &=\exp\left(-\frac{\delta'}{2+\delta'}\times \left(\frac{Xn}{[(1+\delta)\mu]^{t+D}\mu} -
    1\right)[(1+\delta)\mu]^{t+D}\mu/n\right)\\
    &=\exp\left(-\frac{\delta'}{2+\delta'}\times
    \left(\frac{X}{[(1+\delta)\mu]^{t+D}\mu}[(1+\delta)\mu]^{t+D}\mu - [(1+\delta)\mu]^{t+D}\mu/n\right)\right)\\
    &=\exp\left(-\frac{\delta'}{2+\delta'}\times \left(X - [(1+\delta)\mu]^{t+D}\mu/n\right)\right)\\
&=\exp\left(-\frac{\delta'}{2+\delta'}\times \left((1-\delta)^3\mu q \theta - 1 -
[(1+\delta)\mu]^{t+D}\mu/n\right)\right)\\
&=\exp\left(-\frac{\delta'\mu}{2+\delta'}\left((1-\delta)^3 q \theta - 1/\mu -
[(1+\delta)\mu]^{t+D}/n\right)\right)\\
&\leq \exp\left(-\frac{\delta'(1-\delta)^3 q\label{eqn:inequalityA1} \theta\mu}{2(2+\delta')}\right)\\
&\leq\exp\left(-\frac{(1-\delta)^3 q \theta}{6}\mu\right)\label{eqn:inequalityA2}
\end{align}

Since $t+D<\frac{\log n}{(1+\alpha)\log \mu}$, we have
\begin{align}
    t+D\leq \frac{\log n}{\log \mu}\times\frac{1-\frac{\log \mu}{\log
    n}}{1+\frac{\log(1+\delta)}{\log \mu}}.
\end{align}
Hence,
\begin{align}
    \frac{[(1+\delta)\mu]^{t+D}}{n}\leq\frac{1}{\mu}\label{eqn:condition3-1}
\end{align}
Hence, Inequality (\ref{eqn:inequalityA1}) is based on Inequality (\ref{eqn:condition3-1}) and Inequality (\ref{eqn:inequalityA2}) is based on $\delta'\geq 1$ for sufficiently large $n.$
For any node in level $ t + D + 1,$ we have
\begin{align*}
& \Pr\left(\cap_v N_v< X + 1|E_1\right)\\
=& 1 - \Pr \left(\cup_v N_v\geq X + 1|E_1\right)\\
\geq  & 1- \sum_v\Pr \left(N_v\geq X + 1|E_1\right)\\
\geq & 1 - [(1+\delta)\mu]^{t + D + 1} \exp\left(-\Omega(\mu)\right)\\
\geq & 1 -\exp\left((t + D +1)\log[(1+\delta)\mu]-\frac{(1-\delta)^3 q \theta}{6}\mu\right)
\end{align*}

Note we have $t + D\leq \frac{\log n}{(1+\alpha)\log \mu}$. Therefore,
\begin{align*}
&\Pr\left(\cap_v N_v< X + 1|E_1\right)\\
\geq &1- \exp\left(((t+D)\log[(1+\delta)] + (t+D)\log\mu+\log[(1+\delta)\mu])-\frac{(1-\delta)^3 q \theta}{6}\mu\right)\\
\geq &1 - \exp\left(\frac{\log n}{(1+\alpha)\log \mu}\log[(1+\delta)] + \frac{\log n}{(1+\alpha) }+\log[(1+\delta)\mu]-\frac{(1-\delta)^3 q \theta}{6}\mu\right)
\end{align*}
Let $C \leq \frac{6}{(1-\delta)^3 (1+\alpha)}$
and since $\mu>\frac{1}{Cq\theta}\log n,$ we have
\[
\Pr\left(\cap_v N_v< X + 1\right) \geq 1 - \exp\left(-\Omega(\mu)\right)
\]
for sufficiently large $n$. By substituting $X,$ we proved the lemma.
\end{proof}

\section{Proof of $E_4$ and $E_5$}
\subsection{Neighboring structure of all sources}
To prove $E_4$ and $E_5$ happen with a high probability, we first analyze the neighborhood of all
sources in the ER random graph.
In this section, we derived upper and lower bounds of the $t$ neighborhood of all sources. Define
${\cal L}^i_l$ the set of nodes from level $0$ to level $l$ of the BFS tree rooted in source $s_i$.
In addition, define $\phi'_i(v)$ the number of offsprings of node $v$ on the BFS tree rooted in
source $s_i$.
Define
\[
    E^i_1 =\{\forall v \in {\cal L}^i_{t-1}, \phi'_i(v)\in\left( (1-\delta)\mu, (1+\delta)\mu \right)
\]

Denote by the event
\[
    \tilde{E}=\cap_{i=2}^m E^{i}_1 \cap E_1.
\]
We have the following lemma.
\begin{lemma}\label{lem:Etilde}
    If the conditions in Theorem 2 hold, for any $\epsilon>0$,
    \[
        \Pr(\tilde{E})\geq 1-\epsilon
    \]
    for sufficiently large $n$.
\end{lemma}
\begin{proof}
For each infection source $s_i$, follow the similar argument of Lemma 3 in \cite{ZhuYin_15_arxiv}, we have
\[
    \Pr(E^{i}_1)\geq\exp\left(-8\exp\left(-\frac{\delta^2\mu}{2+\delta}+(t-1)\log[(1+\delta)\mu]\right)\right)\geq
    1-\frac{8\left( \mu\left( 1+\delta \right) \right)^{t-1}}{\exp\left(
        \frac{\delta^2\mu}{2+\delta} \right)}
\]
Hence, we have
\[
    \Pr(\bar{E}^{i}_1)\leq \frac{8\left( \mu\left( 1+\delta \right) \right)^{t-1}}{\exp\left(
        \frac{\delta^2\mu}{2+\delta} \right)}
\]
Therefore, with an union bound, we have
\begin{align*}
    \Pr(\cap_{i = 2}^m E^{i}_1) & \geq 1 - \sum_{i=1}^m \Pr(\bar{E}^{i}_1)\\
    & \geq 1 -  \frac{8(m-1)\left( \mu\left( 1+\delta \right) \right)^{t-1}}{\exp\left(
        \frac{\delta^2\mu}{2+\delta} \right)}\\
    & \geq 1- \exp\left(\log 8(m-1) + (t-1)\log\left(\mu(1+\delta)\right)- \frac{\delta^2\mu}{2+\delta} \right)
\end{align*}
To make the probability larger than $1-\epsilon$, we have
\[
    t\leq \frac{\log \frac{\epsilon}{8(m-1)}+\frac{\delta^2\mu}{2+\delta}}{\log (\mu(1+\delta))}+1
\]
Again, we have $t+D<\frac{\log n}{(1+\alpha)\log \mu}$ and $\mu>\frac{1}{Cq\theta}\log n > 3\log n$ which guarantees the
probability goes to 1 asymptotically.

Note the events $\cap_{i =2}^m E^{i}_1$ do not contain the neighborhood of source $s_1$.
Based on Lemma~\ref{lem:E1E2}, with a union bound, it
is straightforward to show that
\[
    \Pr\left(\cap_{i=2}^m E^{i}_1 \cap E_1\right)>1-\epsilon.
\]
for sufficiently large $n$.
Hence the lemma is proved.

\end{proof}

\subsection{Proof of $E_4$ and $E_5$}
\begin{lemma}\label{lem:E4E5}
    If the conditions in Theorem 2 hold, for any $\epsilon>0$,
    \[
        \Pr(E_4, E_5)\geq 1-\epsilon
    \]
    for sufficiently large $n$.
\end{lemma}
\begin{proof}
We still consider the BFS tree rooted at source $s_1$ and we can rewrite the events $E_4$ and $E_5$
in a combined fashion
\[
    E_4\cap E_5 =\{\forall v \in \cup_{i=0}^{t-1} {\cal Z}_i, \psi'(v)\geq (1-\delta)^2\mu q,
    \psi''(v)\geq (1-\delta)^3\mu q\theta
\}.
\]

Define
\[
    {\cal F}_i = \{{\cal Z}_i|\forall v \in{\cal Z}_i, \psi'(v)\geq (1-\delta)^2\mu q,
\psi''(v)\geq (1-\delta)^3\mu q\theta\}
\]
Then, we have
\begin{align*}
 & \Pr(E_4,E_5)\\
   \geq &\Pr(E_4,E_5|\tilde{E})\Pr(\tilde{E})\\
\end{align*}
\begin{align*}
    & \Pr(E_4,E_5 | \tilde{E})\\
  = & \Pr(\forall v \in \cup_{i=0}^{t-1} {\cal Z}_i, \psi'(v)\geq (1-\delta)^2\mu q,\psi''(v)\geq (1-\delta)^3\mu q\theta|\tilde{E})\\
  = & \sum_{{\cal Z}_1\in {\cal F}_1} \cdots\sum_{{\cal Z}_{t - 1}\in {\cal F}_{t - 1}} \Pr(\forall
  v \in {\cal Z}_{t - 1}, \psi'(v)>(1-\delta)^2 \mu q,\psi''(v)\geq (1-\delta)^3\mu q\theta\\
  &| {\cal Z}_{t - 1}, {\cal Z}_{t - 2}, \cdots,
  {\cal Z}_{1},\tilde{E})\Pr( {\cal Z}_{t - 1}, {\cal Z}_{t - 2}, \cdots,
  {\cal Z}_{1}|\tilde{E})
\end{align*}
We have
\begin{align*}
 &   \Pr(\forall v \in {\cal Z}_{t - 1},\psi'(v)>(1-\delta)^2 \mu q| {\cal Z}_{t - 1}, {\cal Z}_{t - 2}, \cdots,
  {\cal Z}_{1},\tilde{E})\\
  \geq & 1 - \sum_{v \in {\cal Z}_{t - 1}} \Pr( \psi'(v)\leq(1-\delta)^2 \mu q,\psi''(v)\geq (1-\delta)^3\mu q\theta| {\cal Z}_{t - 1}, {\cal Z}_{t - 2}, \cdots,
  {\cal Z}_{1},\tilde{E})
\end{align*}
Note conditioned on ${\cal Z}_{t - 1}, {\cal Z}_{t - 2}, \cdots, {\cal Z}_{1}$, consider an
offspring $u$ of node $v$ on the BFS tree rooted at source $s_1$. Node $u$ has two possible states:
infected or susceptible. If $u$ is not infected, $v$ will infect node $u$ with probability $q$ in
the next time slot. On the other hand, if $u$ is infected, it counts as an infected offspring of
node $v$ deterministically. Therefore, $\psi'(v)$ is stochastically lower bounded by binomial
distribution $B( (1-\delta)\mu, q)$. Therefore, with Chernoff bound in \cite{ZhuYin_15_arxiv}, we have
\[
    \Pr( \psi'(v)\leq(1-\delta)^2 \mu q| {\cal Z}_{t - 1}, {\cal Z}_{t - 2}, \cdots,
    {\cal Z}_{1},\tilde{E})\leq \exp\left(-\frac{\delta^2(1-\delta)\mu q}{2}\right)
\]
 Each infected nodes are observed with probability $\theta$ independently. Conditioned on $\psi'(v)\geq (1-\delta)^2 \mu q$, $\psi''(v)$ is stochastically lower bounded by
$B((1-\delta)^2\mu q, \theta)$. Therefore,
\[
    \Pr(\psi''(v)\leq(1-\delta)^3\mu q\theta |\psi'(v)\geq(1-\delta)^2 \mu q, {\cal Z}_{t - 1}, {\cal Z}_{t - 2}, \cdots,
    {\cal Z}_{1},\tilde{E})\leq \exp\left(-\frac{\delta^2(1-\delta)^2\mu q \theta}{2}\right)
\]
Therefore, we have
\begin{align*}
    & \Pr(\psi''(v)\geq(1-\delta)^3\mu q\theta ,\psi'(v)\geq(1-\delta)^2 \mu q| {\cal Z}_{t - 1}, {\cal Z}_{t - 2}, \cdots,
    {\cal Z}_{1},\tilde{E})\\
= &\left(1- \exp\left(-\frac{\delta^2(1-\delta)\mu q}{2}\right)
\right)\left(1 -  \exp\left(-\frac{\delta^2(1-\delta)^2\mu q \theta}{2}\right)\right)\\
\geq & 1 -  \exp\left(-\frac{\delta^2(1-\delta)\mu
q}{2}\right)-\exp\left(-\frac{\delta^2(1-\delta)^2\mu q \theta}{2}\right)\\
\geq & 1 - 2\exp\left(-\frac{\delta^2(1-\delta)^2\mu q \theta}{2}\right)
\end{align*}
Again with union bound, we have
\begin{align*}
 & \Pr(\forall v \in {\cal Z}_{t - 1},\psi'(v)>(1-\delta)^2 \mu q,\psi''(v)\geq(1-\delta)^3\mu q\theta | {\cal Z}_{t - 1}, {\cal Z}_{t -
 2}, \cdots, {\cal Z}_{1},\tilde{E})\\
 \geq & 1 - 2|{\cal Z}_{t - 1}| \exp\left(-\frac{\delta^2(1-\delta)^2\mu q \theta}{2}\right)
\end{align*}
Note, based on event $\tilde{E}$, we have
\[
    |{\cal Z}_{t - 1}| \leq m\sum_{i = 0}^{t - 1}[ (1+\delta)\mu]^i \leq 2m[(1+\delta)\mu]^{t-1}
\]
Hence, we have
\begin{align*}
 & \Pr(\forall v \in {\cal Z}_{t - 1},\psi'(v)>(1-\delta)^2 \mu q,\psi''(v)\geq(1-\delta)^3\mu q\theta| {\cal Z}_{t - 1}, {\cal Z}_{t -
 2}, \cdots, {\cal Z}_{1},\tilde{E})\\
 \geq & 1 - 4m[(1+\delta)\mu]^{t-1} \exp\left(-\frac{\delta^2(1-\delta)^2\mu q \theta}{2}\right)
\end{align*}

Therefore, we have
\begin{align*}
    & \Pr(E_4,E_5 | \tilde{E})\\
  = & \sum_{{\cal Z}_1\in {\cal F}_1} \cdots\sum_{{\cal Z}_{t - 1}\in {\cal F}_{t - 1}} \Pr(\forall
  v \in {\cal Z}_{t - 1}, \psi'(v)>(1-\delta)^2 \mu q,\psi''(v)\geq(1-\delta)^3\mu q\theta\\
  &| {\cal Z}_{t - 1}, {\cal Z}_{t - 2}, \cdots,
  {\cal Z}_{1},\tilde{E}) \Pr( {\cal Z}_{t - 1}, {\cal Z}_{t - 2}, \cdots,
  {\cal Z}_{1}|\tilde{E})\\
  \geq & \sum_{{\cal Z}_1\in {\cal F}_1} \cdots\sum_{{\cal Z}_{t - 1}\in {\cal F}_{t - 1}}
  \left(1-4m[(1+\delta)\mu]^{t-1} \exp\left(-\frac{\delta^2(1-\delta)^2\mu q
  \theta}{2}\right)\right)\\
 \times &\Pr( {\cal Z}_{t - 1}, {\cal Z}_{t - 2}, \cdots,
  {\cal Z}_{1}|\tilde{E})\\
  = &  \left(1 - 4m[(1+\delta)\mu]^{t-1}\exp\left(-\frac{\delta^2(1-\delta)^2\mu q\theta}{2}\right)
\right)\\
\times &\sum_{{\cal Z}_1\in {\cal F}_1} \cdots\sum_{{\cal Z}_{t - 2}\in {\cal F}_{t - 2}} \Pr(\forall
  v \in {\cal Z}_{t - 2}, \psi'(v)>(1-\delta)^2 \mu q,\psi''(v)\geq(1-\delta)^3\mu q\theta\\
  &| {\cal Z}_{t - 2}, {\cal Z}_{t - 3}, \cdots,
  {\cal Z}_{1},\tilde{E}) \Pr( {\cal Z}_{t - 2}, {\cal Z}_{t - 3}, \cdots,
  {\cal Z}_{1}|\tilde{E})\\
\end{align*}

Then, iteratively apply the similar arguments, we obtain,
\begin{align}
    & \Pr(E_4,E_5 | \tilde{E})\\
    \geq & \prod_{i=2}^{t} \left(1 -
    4m[(1+\delta)\mu]^{i-1}\exp\left(-\frac{\delta^2(1-\delta)^2\mu q\theta}{2}\right)\right)\\
    = &  \prod_{i=2}^{t} \left(1 -
  \exp\left(\log 4m + (i-1)\log[(1+\delta)\mu]-\frac{\delta^2(1-\delta)^2\mu q\theta}{2}\right)\right)\\
  \geq &  \prod_{i=2}^{t}\exp\left(-2
  \exp\left(\log 4m + (i-1)\log[(1+\delta)\mu]-\frac{\delta^2(1-\delta)^2\mu q\theta}{2}\right)\right)\\
  = & \exp\left(-2\sum_{i=2}^{t} \exp\left(\log 4m +
  (i-1)\log[(1+\delta)\mu]-\frac{\delta^2(1-\delta)^2\mu q\theta}{2}\right)\right)\\
  = & \exp\left(-2\exp\left(\log 4m -\frac{\delta^2(1-\delta)^2\mu q\theta}{2}\right) \sum_{i = 1}^{t - 1}
  \exp\left( i\log[(1+\delta)\mu]\right)\right)\\
  \geq & \exp\left(-4\exp\left(\log 4m -\frac{\delta^2(1-\delta)^2\mu q\theta}{2}
  +(t-1)\log[(1+\delta)\mu]\right)\right)\label{eqn:E4E5}
\end{align}
To make the probability greater than $1-\epsilon$, we need,
\begin{align}
    t\leq 1 + \frac{\log\log(1-\epsilon)^{-1/4}-\log 2m +\frac{\delta^2(1-\delta)^2\mu q
    \theta}{2}}{\log((1+\delta)\mu)}\label{eqn:anotherT}
\end{align}
Since we have $t+D<\frac{\log n}{(1+\alpha)\log \mu}$ and $\mu>\frac{1}{Cq\theta}\log n > \frac{2}{\delta ^2(1-\delta)^2q\theta}\log n,$ Inequality
\ref{eqn:anotherT} is satisfied when $n$ is large enough.

Therefore, based on Lemma~\ref{lem:Etilde} and Inequality~\ref{eqn:E4E5}, we showed that
\[
    \Pr(E_4,E_5)\geq \Pr(E_4,E_5|\tilde{E})\Pr(\tilde{E})\geq 1-\epsilon
\]
for any $\epsilon>0$ when $n$ is sufficiently large.
\end{proof}
\section{Proof of $E_6$}
\begin{lemma}\label{lem:E6}
    If the conditions in Theorem 2 hold, for any $\epsilon>0$,
    \[
        \Pr(E_6|E_1)\geq 1-\epsilon
    \]
    for sufficiently large $n$.
\end{lemma}

\begin{proof}
Define
\[
        E_7 = \{\tilde{Z}^1_1\geq (1-\delta)^2\mu q\}\cap\{\forall v \in \tilde{Z}_1^1, \cap_{i =
        2}^t\tilde{Z}^i_i(v)\geq(1-\delta)^2\mu q \tilde{Z}^{i-1}_{i-1}(v)\}
\]
Following the similar arguments in Lemma 5 in \cite{ZhuYin_15_arxiv}, we have
\[
    \Pr(E_7|E_1)\geq 1-\epsilon.
\]

Note, for each node $v\in \tilde{Z}_1^1$, based on event $E_7$, we have $\tilde{Z}_t^t(v)\geq[(1-\delta)^2\mu
q]^{t-1}$. Recall each infected node report its status independently. Therefore, $\tilde{Z}'^t_t(v)$ is stochastically lower bounded
by $\hbox{Bi}([(1-\delta)^2\mu q]^{t-1},\theta)$. By Chernoff bound, we have
\[
    \Pr( \tilde{Z}'^t_t(v)\geq [(1-\delta)^2\mu q]^{t-1}(1-\delta)\theta|E_7, E_1)\leq
    \exp\left(-\frac{\delta^2[(1-\delta)^2\mu q]^{t-1}\theta}{2}\right)
\]
Note $\tilde{Z}^1_1\leq (1+\delta)\mu$ based on event $E_1$, with a union bound, we have
\begin{align*}
&    \Pr(\forall v \in \tilde{\cal Z}^1_1,  \tilde{Z}'^t_t(v)\geq [(1-\delta)^2\mu q]^{t-1}(1-\delta)\theta|E_7, E_1)\\
\geq & 1 -   (1+\delta)\mu\exp\left(-\frac{\delta^2[(1-\delta)^2\mu q]^{t-1}\theta}{2}\right)\\
\geq & 1- \exp\left(\log[ (1+\delta)\mu]-\frac{\delta^2[(1-\delta)^2\mu q]^{t-1}\theta}{2}\right)\\
\geq & 1 - \epsilon,
\end{align*}
for sufficiently large $n$ since $\mu>\frac{1}{Cq\theta}\log n$.

Therefore, we have
\[
    \Pr(E_6|E_1)\geq \Pr(E_6|E_7,E_1)\Pr(E_7|E_1)\geq 1-\epsilon.
\]
The lemma is proved.
\end{proof}

\end{document}